\documentclass[11pt]{article}
\usepackage{fullpage}
\usepackage{datetime}
\settimeformat{ampmtime}
\usdate
\usepackage{scrtime}
\usepackage[draft,scrtime]{prelim2e}
\renewcommand{\PrelimWords}{DRAFT}
\usepackage{amsmath, amsthm, amssymb}
\usepackage{float}
\usepackage{cite}
\usepackage{graphicx}
\usepackage{subfigure}
\usepackage{epstopdf}
\usepackage{multirow}
\usepackage{url,color}
\usepackage{hyperref, color} 

\newtheorem{thm}{Theorem}
\newtheorem{lem}{Lemma}
\newtheorem{cor}{Corollary}

\newtheorem{prop}{Proposition}

\allowdisplaybreaks

\newcommand{\R}{\mathbb{R}}

\newcommand{\E}{\operatorname{E}}
\newcommand{\F}{\mathrm{F}}

\newcommand{\e}{\mathrm{e}}
\renewcommand{\j}{\iota}

\newcommand{\vct}[1]{\boldsymbol{#1}}
\newcommand{\mtx}[1]{\boldsymbol{#1}}

\newcommand{\<}{\langle}
\renewcommand{\>}{\rangle}
\newcommand{\T}{\mathrm{T}}

\newcommand{\set}[1]{\mathcal{#1}}

\newcommand{\va}{\vct{a}}

\newcommand{\vd}{\vct{d}}
\newcommand{\ve}{\vct{e}}

\newcommand{\vh}{\vct{h}}

\newcommand{\vp}{\vct{p}}
\newcommand{\vq}{\vct{q}}

\newcommand{\vu}{\vct{u}}
\newcommand{\vv}{\vct{v}}
\newcommand{\vw}{\vct{w}}
\newcommand{\vx}{\vct{x}}
\newcommand{\vy}{\vct{y}}
\newcommand{\vz}{\vct{z}}

\newcommand{\vxi}{\vct{\xi}}
\newcommand{\veta}{\vct{\eta}}
\newcommand{\mA}{\mtx{A}}
\newcommand{\mB}{\mtx{B}}
\newcommand{\mC}{\mtx{C}}
\newcommand{\mD}{\mtx{D}}

\newcommand{\mF}{\mtx{F}}

\newcommand{\mH}{\mtx{H}}
\newcommand{\mI}{\mtx{I}}

\newcommand{\mL}{\mtx{L}}

\newcommand{\mQ}{\mtx{Q}}
\newcommand{\mR}{\mtx{R}}
\newcommand{\mS}{\mtx{S}}

\newcommand{\mU}{\mtx{U}}
\newcommand{\mV}{\mtx{V}}
\newcommand{\mW}{\mtx{W}}
\newcommand{\mX}{\mtx{X}}
\newcommand{\mY}{\mtx{Y}}
\newcommand{\mZ}{\mtx{Z}}

\newcommand{\mPhi}{\mtx{\Phi}}

\newcommand{\mSigma}{\mtx{\Sigma}}

\newcommand{\setA}{\set{A}}
\newcommand{\setB}{\set{B}}

\newcommand{\setO}{\set{O}}

\def\l{\ell}

\def\PT{\mathcal{P}_{T}}
\def\PTc{\mathcal{P}_{T^\perp}}
\def\O{\Omega}
\def\o{\omega}

\def\PP{\mathbb{P}}

\begin{document}
	

\title{Compressive Sampling of Ensembles of Correlated Signals}

\author{Ali Ahmed and Justin Romberg\thanks{A.\ A.\ is with the Department of Mathematics at MIT in Cambridge, Massachusetts and J.\ R.\ is with the School of Electrical and Computer Engineering at Georgia Tech in Atlanta, Georgia.  Email: alikhan@mit.edu, jrom@ece.gatech.edu.  This work was supported by NSF grant CNS-0910592, ONR grant N00014-11-1-0459, and a grant from the Packard Foundation.}
}

\date{DRAFT: \currenttime, \today}

\renewcommand{\PrelimWords}{Draft by A. Ahmed and J. Romberg}

\maketitle

\begin{abstract}
We propose several sampling architectures for the efficient acquisition of an ensemble of correlated signals.  We show that without prior knowledge of the correlation structure, each of our architectures (under different sets of assumptions) can acquire the ensemble at a sub-Nyquist rate.  Prior to sampling, the analog signals are diversified using simple, implementable components.  The diversification is achieved by injecting types of ``structured randomness'' into the ensemble, the result of which is subsampled.  For reconstruction, the ensemble is modeled as a low-rank matrix that we have observed through an (undetermined) set of linear equations.  Our main results show that this matrix can be recovered using a convex program when the total number of samples is on the order of the intrinsic degree of freedom of the ensemble --- the more heavily correlated the ensemble, the fewer samples are needed.

To motivate this study, we discuss how such ensembles arise in the context of array processing.

\end{abstract}


\section{Introduction}

This paper considers the exact reconstruction of correlated signals from the samples collected at a sub-Nyquist rate. We propose several implementable architectures, and derive a sampling theorem that relates the bandwidth and the (a priori unknown) correlation structure to the sufficient sampling rate for successful signal reconstruction.

We consider ensembles of signals output from $M$ sensors, each of which is bandlimited to frequencies below $W/2$ (see Figure~\ref{fig:signal-ensemble}).  The entire ensemble can be acquired by taking $W$ uniformly spaced samples per second in each channel, leading to a combined sampling rate of $MW$.  We will show that if the signals are correlated, meaning that the ensemble can be written as (or closely approximated by) distinct linear combinations of $R\ll M$ latent signals, then this net sampling rate can be reduced to approximately $RW$ using {\em coded acquisition}.  The sampling architectures, we propose are blind to the correlation structure of the signals; this structure is discovered as the signals are reconstructed.

Each architecture involves a different type of {\em analog diversification} which ensures that the signals are sufficiently ``spread out'' so each point sample captures information about the ensemble.  Ultimately, what is measured are not actual samples of the individual signals, but rather are different linear combinations that combine multiple signals and capture information over an interval of time.  Later, we will show that these samples can be expressed as linear measurements of a low-rank matrix. Over the course of one second, we aim to acquire an $M\times W$ matrix comprised of samples of the ensemble taken at the Nyquist rate. The proposed sampling architecture produces a series of linear combinations of entries of this matrix.  Conditions under which a low-rank matrix can be effectively recovered from an under-determined set of linear measurements have been the object of intense study in the recent literature \cite{fazel02ma,recht10gu,candes09ex,gross11re}; the mathematical contributions in this paper show how these conditions are met by systems with clear implementation potential.

Our motivation for studying these architectures comes from classical problems in array signal processing.  In these applications, one or more ``narrowband'' signals are measured at multiple sensors at different spatial locations.  While narrowband signals can have significant bandwidth, they are modulated up to a high carrier frequency, making them very heavily spatially correlated as they arrive at the array.  This correlation, which we review in more detail in Section~\ref{sec:APP}, can be systematically exploited for spatial filtering (beamforming), interference removal, direction-of-arrival estimation, and multiple source separation.  These activities all depend on estimates of the inter-sensor correlation matrix, and the rank of this matrix can typically be related to the number of sources that are present.

Compressive sampling has been used in array processing in the past: sparse regularization was used for direction of arrival estimation \cite{gorodnitsky97sp,fuchs99mu,fuchs01ap} long before any of the ``sub-Nyquist'' sampling theorems started to make the theoretical guarantees concrete \cite{candes06ro,kunis06ra,rudelson08sp}.  These results (along with more recent works including \cite{duarte13sp,tang13co,candes14to}), show how exploiting the structure of the array response in free space (for narrowband signals, this consists of samples of a superposition of a small number of sinusoids) can be used to either super-resolve the DOA estimate or reduce the number of array elements required to locate a certain number of sources.  A single sample is associated with each sensor, and the acquisition complexity scales with the number of array elements.

In this paper, we exploit this structure in a different way.  Our goal is to completely reconstruct the time-varying signals at all the array elements.  The structure imposed on this ensemble is more general than the spatial spectral sparsity in the previous work; we ask that the signals are correlated in some a priori unknown manner.  Our ensemble sampling theorems remain applicable even when the array response depends on the position of the source in a complicated way.  Moreover, our reconstruction algorithms are indifferent to what the spatial array response actually is, as long as the narrowband signals remain sufficiently correlated.

The paper is organized as follows. In Sections~\ref{sec:signal-model} and \ref{sec:APP} we describe the signal model and its motivation from problems in array processing.  In Section~\ref{sec:Arch-comps}, we introduce the components (and their corresponding mathematical models) that we will use in our sampling architectures.  In Section~\ref{sec:Samp-Archs}, we present the sampling architectures, show how the measurements taken correspond to generalized measurements of a low-rank matrix, and state the relevant sampling theorems.  Numerical simulations, illustrating our theoretical results, are presented in Section \ref{sec:Exps}. Finally, Section \ref{sec:Proof-Exact-rec}, and Section \ref{sec:Stability} provide the derivation of the theoretical results. 

%
%
%

\subsection{Notation}
We use upper, and lower case bold letters for matrices and vectors, respectively. Scalars are represented by upper, and lower case, non-bold letters. The notation $\vx^
*$  denotes a row vector formed by taking the hermitian transpose of a column vector $\vx$. Linear operators, and sets are represented using script letters. We use $[N]$ to denote the set $\{1,2,3,\ldots,N\}$. The notation $\mI_W$ denotes a $W\times W$. For a set $\setB_\ell \subset [W]$,  $\mI_{\setB_\ell}$ denotes a $W \times W$ matrix with ones at diagonal positions indexed by $\setB_\ell$, and zeros elsewhere. Given two matrices $\mA$, and $\mB$, we denote by $\mA \boxtimes \mB$, the rank-1 matrix: $[\mbox{vec}(\mA)][\mbox{vec}(\mB)]^\T$, where $\mbox{vec}(\mA)$, and $\mbox{vec}(\mB)$ are the column vectors formed by stretching the columns of $\mA$, and $\mB$, respectively, and $\T$ denotes the transpose. We will use $\mA \otimes \mB$ is the usual Kronecker product of $\mA$, and $\mB$. We will use $\mathbf{1}_P$ to denote a $P \times 1$ vector of all ones. Lastly, the operator $\E$ refers to the expectation operator, and $\mathbb{P}$ represents the probability measure. 

\subsection{Signal model}
\label{sec:signal-model}

Our signal model is illustrated in Figure~\ref{fig:signal-ensemble}.  We denote a continuous-time signal ensemble by $\mX_c(t)$: a set of $M$ individual signals  $x_1(t),\ldots,x_M(t)$.  Conceptually, we may think of $\mX_c(t)$ as a ``matrix'' with a finite number $M$ of rows, with each row containing a bandlimited signal.  Our underlying assumption is that every signal in the ensemble can be approximated as the linear combination of underlying $R$ independent signals in a smaller ensemble $\mS_c(t)$. We write  
\begin{equation}
\label{eq:lowrankensemble}
\mX_c(t) = \mA\mS_c(t),
\end{equation}
where $\mA$ is an $M\times R$ matrix with entries $A[m,r]$.  We will use the convention that fixed matrices operating to the left of the signal ensembles simply ``mix'' the signals point-by-point, and so \eqref{eq:lowrankensemble} is equivalent to
\[
x_m(t) = \sum_{r=1}^R A[m,r]s_r(t).
\]

\begin{figure}
	\centering
	\begin{tabular}{cc}
		\includegraphics[height=1.2in]{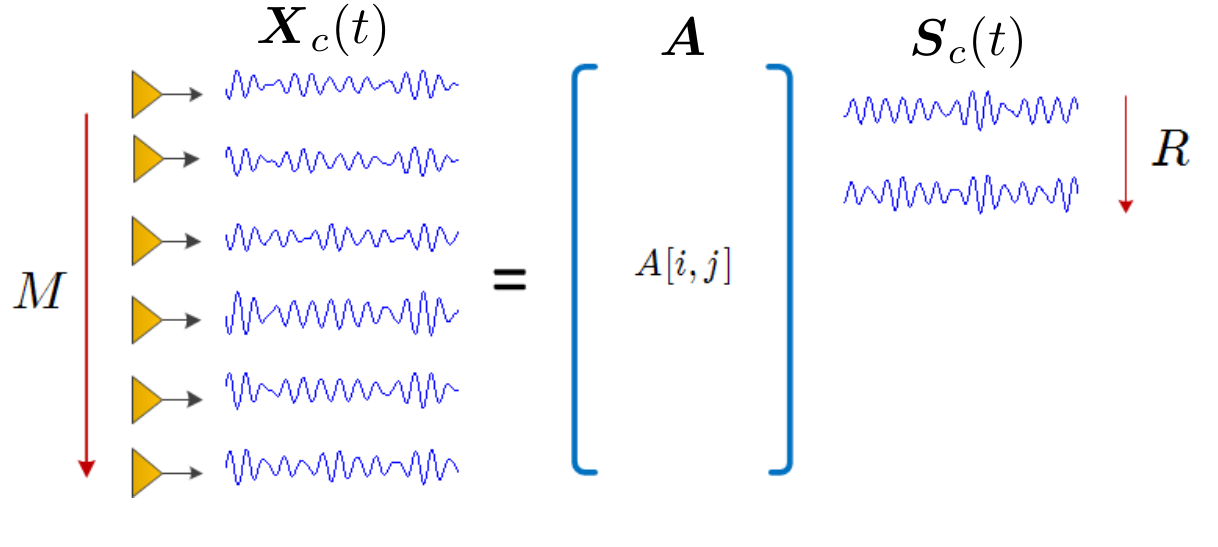} &
		\includegraphics[height=1.2in]{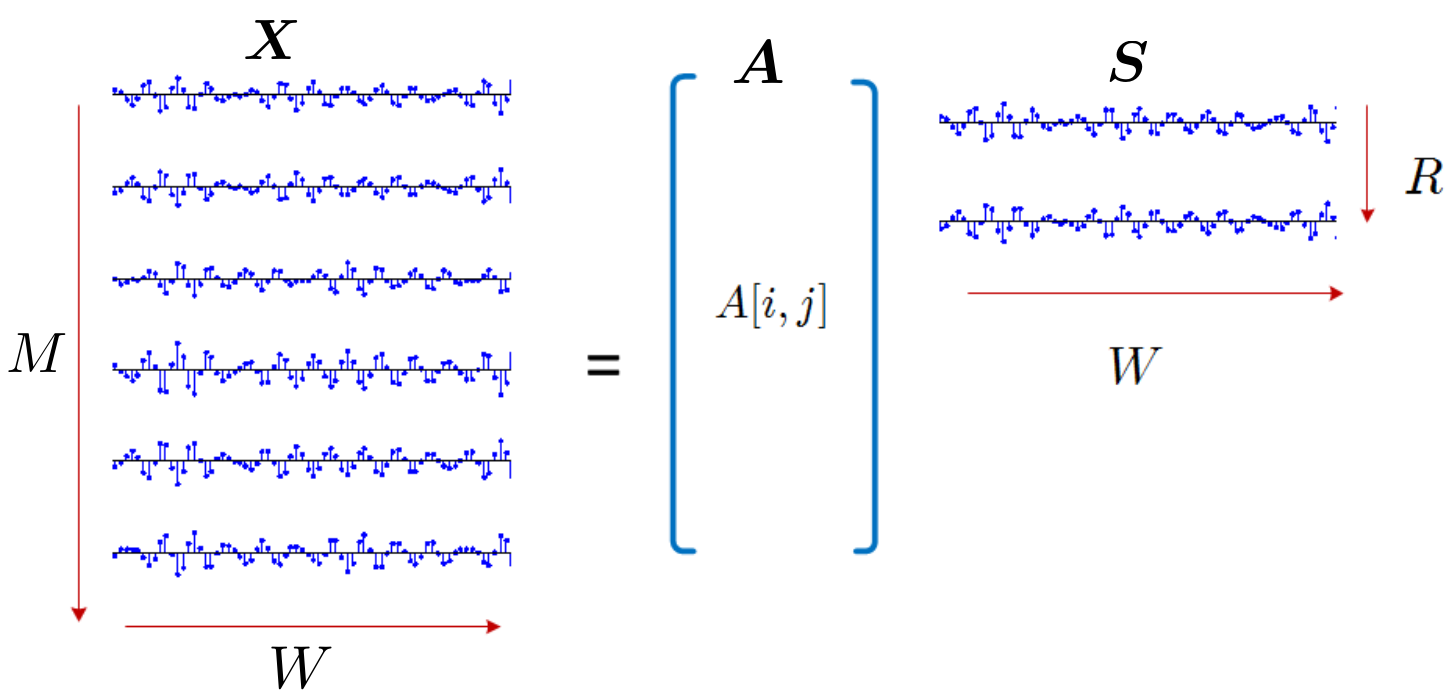} \\
		(a) & (b) 
	\end{tabular}
	\caption{\small\sl (a) Our model is that an ensemble of continuous-time signals are {\em correlated}, meaning the $M$ signals can be closely approximated by a linear combination of $R$ underlying signals.  We can write the $M$ signals in $\mX_c(t)$ as a tall matrix (capturing the correlation structure) multiplied by an ensemble of $R$ latent signals.  (b) The matrix of samples inherits the low-rank structure of the continuous-time ensemble.}
	\label{fig:signal-ensemble}
\end{figure} 
The only structure, we impose on the individual signals is that they are real-valued, and bandlimited.  To keep the mathematics clean, we take the signals to be periodic for now, however, the results can be extended to non-periodic signals as will be discussed shortly. We begin with a natural way to discretize the problem; that is, what exists in $\mX_c(t)$ for $t\in [0,1)$ is all there is to know, and each signal can be captured exactly with $W$ equally-spaced samples.  Each bandlimited, periodic signal in the ensemble can be written as
\begin{equation}\label{eq:Signaltx}
x_m(t) = \sum_{\o=-B}^B\alpha_m[\o]\, \e^{\j 2\pi \o t},
\end{equation}
where the $\alpha_m[\o]$ are complex but are symmetric, $\alpha_m[-\o]=\alpha_m[\o]^*$, to ensure that $x_m(t)$ is real.  We can capture $x_m(t)$ perfectly by taking $W=2B+1$ equally spaced samples per row.  We will call this the $M\times W$ matrix of samples $\mX$; knowing every entry in this matrix is the same as knowing the entire signal ensemble.  We can write
\begin{equation}\label{eq:Cdef}
\mX = \mC\mF^*,
\end{equation}
where $\mC$ is an $M\times W$ matrix whose rows contain Fourier series coefficients for the signals in $\mX_c(t)$, and $\mF$ is a $W\times W$ normalized discrete Fourier matrix with entries
\[
F[\omega,\ell] = \frac{1}{\sqrt{W}} \e^{-\j 2\pi \omega \ell/W},\quad  \omega = -B,-B+1,\ldots, B, ~~\ell = 0,1,2,\ldots, W-1.
\]
Observe that both $\mX$, and hence $\mC$ inherit the correlation structure of the ensemble $\mX_c(t)$.  Before moving on, observe that \eqref{eq:lowrankensemble}, and \eqref{eq:Signaltx} impose an $R$, and $B$ dimensional subspace structure on $\mX$, where $\mbox{rank}(\mX) = \min(R,B+1)$. If $R \geq B+1$ then we can take $R = B+1$ with the underlying independent signals in $\mS_c(t)$ being the known sinusoids at frequencies  $\omega = 0, 1, 2, \ldots , B$. However, we are interested in the more pertinent and challenging case of $R < B+1$. In this case, the underlying independent signals in $\mS_c(t)$ are not known in advance, and the main contribution of this paper is to leverage this unknown correlation structure in $\mX_c(t)$ to reduce the sampling rate. Lastly, in the interest of readability of our technical results, we assume without loss of generality that $W \geq M$, that is, the bandwidth of the signals is greater than the number of signals. 

Same correlated signal model was considered in \cite{ahmed2015compressive} for compressive sampling of multiplexed signals. Two multiplexing architectures were proposed and for each, a sampling theorem was proved that dictated minimum number of samples for exact recovery of the signal ensemble. This paper presents sampling architectures, where we use a separate ADC for each channel and rigorously prove that ADCs can operate at roughly the optimal sampling rate to guarantee signal recovery. Other types of correlated signal models have been exploited previously to achieve gains in the sampling rate. For example, \cite{hormati2010distributed} shows that two signals related by a sparse convolution kernel can be reconstructed jointly at a reduced sampling rate. The signal model in \cite{baron2009distributed} considers multiple signals residing in a fixed subspace spanned by a subset of the basis functions of a known basis, and shows that the sampling rate to successfully recover the signals scales with the number of basis functions used in the construction of the signals. In this paper, we also show that the sampling rate scales with the number of independent latent signals but we do this without the knowledge of the basis. For a more applied treatment of the results with similar flavor as in \cite{baron2009distributed}, we refer the reader to \cite{mishali11xasub,mishali2009blind,mishali2011xampling}.

As will be shown later, we observe the signal ensemble $\mX_c(t)$ through a limited set of random projections, and signal recovery is achieved by a nuclear norm minimization program. A related work \cite{mantzel2014compressed} considers the case when given a few random projections of a signal, we find out the subspace to which it belongs by solving a series of least-squares programs. 

\subsection{Extension to non-periodic signals}
We end this section by noting that their are many ways this problem might be discretized.  Using Fourier series is convenient in two ways: we can easily tie together the notion of a signal being bandlimited with having a limited support in Fourier space, and our sampling operators have representations in Fourier space that make them more straightforward to analyze.  In practice, however, the recovery technique can be extended to non-periodic signals by windowing the input, and representing each finite interval using any one of a number of basis expansions --- the low rank structure is preserved under any linear representation.  It is also possible that we are interested in performing the ensemble recovery over multiple time frames, and would like the recovery to transition smoothly between these frames.  For this we might consider a windowed Fourier series representations (e.g.\ the lapped orthogonal transform in \cite{malvar89lo}) that are carefully designed so that the basis functions are tapered sinusoids (so we again get something close to bandlimited signals by truncating the representation to a certain depth) but remain orthonormal.  It is also possible to adjust our recovery techniques to allow for measurements which span consecutive frames, yielding another natural way to tie the reconstructions together.

A framework similar to this for sparse recovery is described in detail in \cite{asif2014sparse}.
\subsection{Applications in array signal processing}
\label{sec:APP}

One application area where low-rank ensembles of signals play a central role is array processing of narrowband signals.  In this section, we briefly review how these low-rank ensembles arise.  The central idea is that sampling a wavefront at multiple locations in space (as well as in time) leads to redundancies which can be exploited for spatial processing.  These concepts are very general, and are common to applications as diverse as surveillance radars, underwater acoustic source localization and imaging, seismic exploration, wireless communications.

The essential scenario is that multiple signals are emitted from different locations, each of the signals occupies the same bandwidth of size $W$ which has been modulated up to a carrier frequency $\omega_c$.  The signals observed by receivers in the array are, to a rough approximation, complex multiples of one another.  To a very close approximation, the observed signals lie in a subspace with dimension close to one --- this subspace is determined by the location of the source.  This redundancy between the observations at the array elements is precisely what causes the ensemble of signals to be low rank; the rank of the ensemble is determined by the number of emitters.  The only conceptual departure from the discussion in previous sections, as we will see below, is that each emitter may be responsible for a subspace spanned by a number of latent ``signals'' that is greater than one (but still small).

Having an array with a large number of appropriately spaced elements can be very advantageous even when there only a relatively small number of emitters present.  Observing multiple delayed versions of a signal allows us to perform spatial processing,  we can beamform to enhance or null out emitters at certain angles, and separate signals coming from different emitters.  The resolution to which we can perform this spatial processing depends on the number of elements in the array (and their spacing).

The main results of this paper do not give any guarantees about how well these spatial processing tasks can be performed.  Rather, they say that the same correlation structure that makes these tasks possible can be used to lower the net sampling rate over time.  The entire signal ensemble can be reconstructed from this reduced set of samples, and spatial processing can follow.

We now discuss in more detail how these low rank ensembles come about.  For simplicity, this discussion will center on linear arrays in free space.  As we just need the signal ensemble to lie in a low dimensional subspace, and do not need to know what this subspace may be beforehand, the essential aspects of the model extend to general array geometries channel responses and frequency-selective/multipath channels.

Suppose that a signal is incident on the array (as a plane wave) at an angle $\theta$.  Each array element observes a different shift of this signal --- if we denote what is seen at the array center (the origin in Figure~\ref{fig:arrayprocessing}(a)) by $s(t)$, then an element $m$ at distance $d_m$ from the center sees $x_m(t) = s(t-(d_m/c)\sin\theta)$.  If the signal consists of a single complex sinusoid, $s(t) = \e^{\j 2\pi \omega t}$, then these delays translated into different (complex) linear multiples of the same signal,
\begin{equation}
\label{eq:narrowsingle}
x_m(t) = \e^{-\j 2\pi \omega r_m\sin(\theta)/c}\,\e^{\j 2\pi \omega t}.
\end{equation}
In this case, the signal ensemble has rank\footnote{We are using complex numbers here to make the discussion go smoothly; the real part of the signal ensemble is rank $2$, having a $\cos(2\pi \omega t)$ and a $\sin(2\pi \omega t)$ term.} $1$; we can write $\mX_c(t) = \va(\theta,\omega)\e^{\j 2\pi \omega t}$, where $\va(\theta,\omega)$ is an $M$-dimensional {\em steering} vector of complex weights given above.

This decomposition of the signal ensemble makes it clear how spatial information is coded into the array observations.  For instance, standard techniques \cite{schmidt1986multiple,roy1989esprit} for estimating the direction of arrival involve forming the spatial correlation matrix by averaging in time,
\[
\mR_{xx} = \frac{1}{L}\sum_{\ell=1}^L \mX(t_\ell)\mX(t_\ell)^*.
\] 
As the column space of $\mR_{xx}$ should be $\va(\theta,\omega)$, we can correlate the steering vector for every direction to see which one comes closest to matching the principal eigenvector of $\mR_{xx}$.

The ensemble remains low rank when the emitter has a small amount of bandwidth relative to a larger carrier frequency.  If we take $s(t) = s_b(t)\,\e^{\j2\pi \omega_c t}$, where $s_b(t)$ is bandlimited to $W/2$, then when $W \ll \omega_c$, the $\va(\theta,\omega)$ for $\omega\in[\omega_c-W/2,\omega_c+W/2]$ will be very closely correlated with one another.  In the standard scenario where the array elements are uniformly spaced $c/(2\omega_c)$ along a line, we can make this statement more precise using classical results on spectral concentration \cite{slepian76ba,slepian78pr}.  In this case, the steering vectors $\va(\theta,\omega)$ for $\omega\in [\omega_c\pm W/2]$ are equivalent to integer spaced samples of a signal whose (continuous-time) Fourier transform is bandlimited to frequencies in $(1\pm W/(2\omega_c))(\sin\theta)/2$, for a bandwidth less than $W/(2\omega_c)$.  Thus the dimension of the subspace spanned by $\{\va(\theta,\omega),~\omega\in[\omega_c\pm W/2]\}$ is, to within a very good approximation, $\approx MW/\omega_c + 1$.

Figure~\ref{fig:arrayprocessing}(b) illustrates a particular example.  The plot shows the (normalized) eigenvalues of the matrix
\begin{equation}
\label{eq:Raa}
\mR_{aa} = \int_{\omega_c-W/2}^{\omega_c+W/2} \va(\theta,\omega)\va(\theta,\omega)^*\, d\omega,
\end{equation}
for the fixed values of $\omega_c = 5$ GHz, $W = 100$ MHz, $c$ equals the speed of light, $M=101$, and $\theta=\pi/4$.  We have  $MW/\omega_c + 1 = 3.02$, and only $3$ of the eigenvalues are within a factor of $10^4$ of the largest one.

It is fair, then, to say that the rank of the signal ensemble is a small constant times the number of narrow band emitters.

\begin{figure}
	\centering
	\begin{tabular}{cc}
		\includegraphics[height=2in]{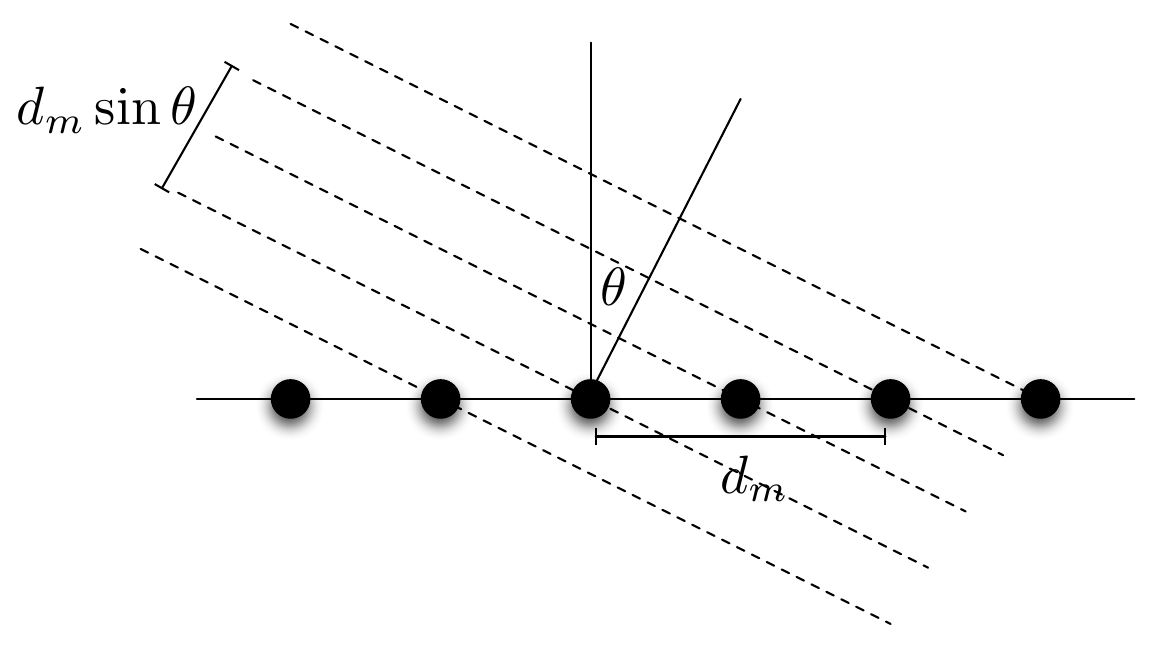} \hspace*{0.1in} &		
		\raisebox{0.2in}{\rotatebox{90}{\small $\log_{10}(k$th largest eigenvalue$)$}}
		\includegraphics[height=2in]{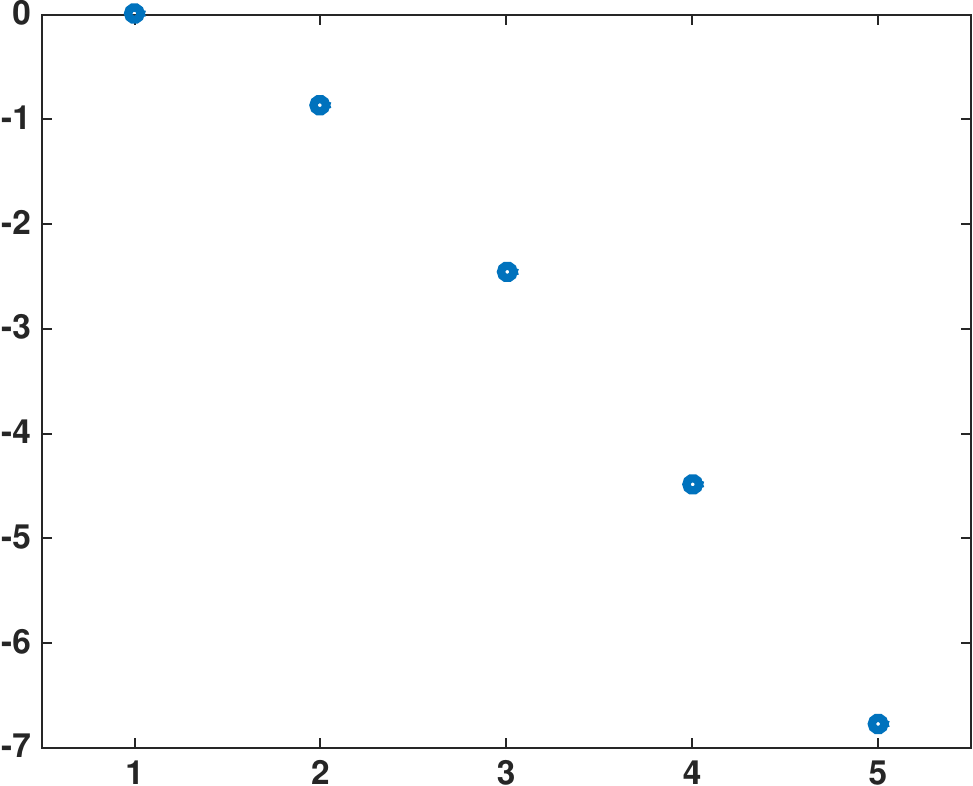} \\[-2mm]
		& $k\rightarrow$ \\
		(a) & (b) 
	\end{tabular}
	\caption{\small\sl (a) A plane wave impinges on a linear array in free space.  When the wave is a pure tone in time, then the responses at each element will simply be phase shifts of one another.  (b) Eigenvalues for $\mR_{aa}$, on a $\log_{10}$ scale and normalized so that the largest eigenvalue is $1$, defined in \eqref{eq:Raa} for an electromagnetic signal with a bandwidth of $100$ MHz and a carrier frequency of $5$ GHz; the array elements are spaced half a carrier-wavelength apart.  Even when the signal has an appreciable bandwidth, the signals at each of the array elements are heavily correlated --- the effective dimension in this case is $R=3$ or $4$.}
	\label{fig:arrayprocessing}
\end{figure}

\subsection{Architectural components}
\label{sec:Arch-comps}

In addition to analog-to-digital converters, our proposed architectures will use three standard components:  analog vector-matrix multipliers, modulators, and linear time-invariant filters.  The signal ensemble is passed through these devices, and the result is sampled using an analog-to-digital converter (ADC) taking either uniformly or non-uniformly spaced samples --- these samples are the final outputs of our acquisition architectures.

\begin{figure}[htp]
	\begin{center}
		\includegraphics[trim = 2cm 10.5cm 2cm 0cm, scale = 0.6 ]{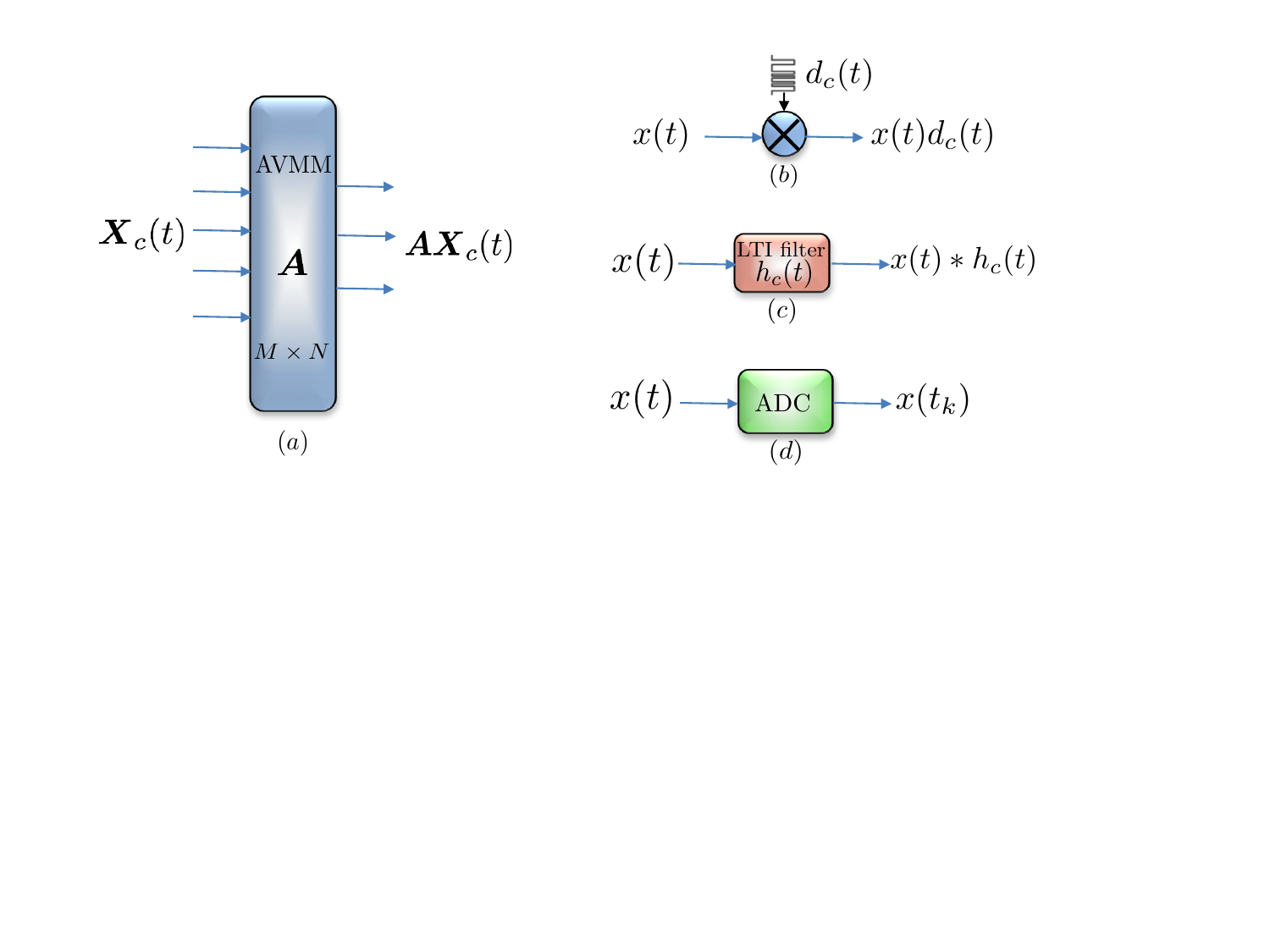}
	\end{center}
	\caption{\small\sl (a) The analog vector-matrix multiplier (AVMM) takes random linear combinations of $M$ input signals to produce $N$ output signals. The action of AVMM can be thought of as the left multiplication of random matrix $\mA$ to ensemble $\mX_c(t)$. Intuitively, this operation amounts to distributing energy in the ensemble equally across channels. (b) Modulators multiply a signal in analog with a random binary waveform that disperses energy in the Fourier transform of the signal. (c) Random LTI filters randomize the phase information in the Fourier transform of a given signal by convolving it with $h_c(t)$ in analog, which distributes energy in time. (d) Finally, ADCs convert an analog stream of information in discrete form. We use both uniform and non-uniform sampling devices in our architectures.}
	\label{fig:Fig3}
\end{figure}

The analog vector-matrix multiplier (AVMM) produces an output signal ensemble $\mA\mX_c(t)$ when we input it with signal ensemble $\mX_c(t)$, where $\mA$ is an $N\times M$ matrix whose elements are fixed.  Since the matrix operates pointwise on the ensemble of signals, sampling output $\mA\mX_c(t)$ is the same as applying $\mA$ to matrix $\mX$ of the samples (i.e., sampling commutes with the application of $\mA$).  Recently, AVMM blocks have been built with hundreds of inputs and outputs and with bandwidths in the tens-to-hundreds of megahertz \cite{schlottmann11hi, chawla04a531}. We will use the AVMM block to ensure that energy disperses more or less evenly throughout the channels.  If $\mA$ is a random orthogonal transform, it is highly probable that each signal in $\mA\mX_c(t)$ will contain about the same amount of energy regardless of how the energy is distributed among the signals in $\mX_c(t)$ (formalized in Lemma~\ref{lem:coherence} below), allowing us to deploy equal sampling resources in each channel while ensuring that resources on ``quiet'' channels  are not being wasted. 

The second component of the proposed architecture is the modulators, which simply take a single signal $x(t)$ and multiply it by fixed and known signal $d_c(t)$.  We will take $d_c(t)$ to be a binary $\pm 1$ waveform that is constant over time intervals of a certain length $1/W$. That is, the waveform alternates at the Nyquist sampling rate. If we take $W$ samples of $d_c(t)x(t)$ on $[0,1)$, then we can write the vector of samples $\vy$ as
\begin{equation}\label{eq:modulator-samples}
\vy = \mD\vx,
\end{equation}
where $\vx$ is the $W$-vector containing the Nyquist-rate samples of $x(t)$ on $[0,1)$, and $\mD$ is an $W\times W$ diagonal matrix whose non-zero entries are  samples $\vd \in\R^W$ of $d_c(t)$. We will choose a binary sequence that randomly generates $d_c(t)$, which amounts to $\mD$ being a random matrix of the following form:
\begin{equation}
\label{eq:modulatorD}
\mD = 
\begin{bmatrix}
d[0] & & & \\
& d[1] & & \\
& & \ddots & \\
& & & d[W-1]
\end{bmatrix}
\quad 
\text{where $d[n]= \pm 1$ with probability $1/2$}, 
\end{equation}
and the $d[n]$ are independent.
Conceptually, the modulator disperses the information in the entire band of $x(t)$ --- this allows us to acquire the information at a smaller rate by filtering a sub-band as will be shown in Section \ref{sec:Samp-Archs}. 

Compressive sampling architectures based on the random modulator have been analyzed previously in the literature \cite{tropp2010beyond,mishali2009blind}.  The principal finding is that if the input signal is spectrally sparse (meaning the total size of the support of its Fourier transform is a small percentage of the entire band), then the modulator can be followed by a low-pass filter and an ADC that takes samples at a rate comparable to the size of the active band.  This architecture has been implemented in hardware in multiple applications \cite{laska07th,yoo12a100,yoo12co,mishali11xasub,murray11de}.

The third type of component we will use to preprocess the signal ensemble is a linear time-invariant (LTI) filter that takes an input $x(t)$ and convolves it with a fixed and known impulse response $h_c(t)$.  We will assume that we have complete control over $h_c(t)$, even though this brushes aside admittedly important implementation questions.  Because $x(t)$ is periodic and bandlimited, we can write the action of the LTI filter as a $W\times W$ circular matrix $\mH$ operating on samples $\vx$ (the first row of $\mH$ consists of samples $\vh$ of $h_c(t)$), that is, $\vy = \mH\vx$, where $\vy$ is the vector of $W$ samples in $ t \in [0,1)$ of the signal obtained at the output of the filter. We will make repeated use of the fact that $\mH$ is diagonalized by the discrete Fourier transform:
\begin{equation}
\label{eq:filter1}
\mH = \mF^*\hat{\mH}\mF,
\end{equation}
where $\mF$ is the $W\times W$ normalized discrete Fourier matrix with entries, and $\hat{\mH}$ is a diagonal matrix whose entries are $\hat{\mH} = \sqrt{W}\mF\mH$.  The vector $\hat{\mH}$ is a scaled version of the non-zero Fourier series coefficients of $h_c(t)$.

To generate the impulse response, we will use a random unit-magnitude sequence in the Fourier domain.  In particular, we will take 
\begin{equation}
\label{eq:filter2}
\hat{\mH} =
\begin{bmatrix}
\hat{h}(0) & & & \\
& \hat{h}(1) & & \\
& & \ddots & \\
& & & \hat{h}(W-1)
\end{bmatrix}
\end{equation}
where
\begin{equation}
\label{eq:filter3}
\hat{h}(\omega) = 
\begin{cases}
\pm 1,\text{with prob.\ $1/2$}, & \omega = 0 \\
e^{j\theta_\omega},~\text{with}~\theta_\omega\sim\mathrm{Uniform}([0,2\pi]), & 1\leq \omega \leq (W-1)/2 \\
\hat{h}(W-\omega)^*, & (W+1)/2\leq\omega\leq W-1
\end{cases}.
\end{equation}
These symmetry constraints are imposed so that $\vh$ (and hence, $h_c(t)$) is real-valued.  Conceptually, convolution with $h_c(t)$ disperses a signal over time while maintaining fixed energy (note that $\mH$ is an orthonormal matrix).

Convolution with a random pulse followed by sub-sampling has also been analyzed in the compressed sensing literature \cite{romberg2009compressive,haupt10to,rauhut12re,tropp2006random}.  If the random filter is created in the Fourier domain as above, then following the filter with an ADC that samples at random locations produces a universally efficient compressive sampling architecture --- the number of samples that we need to recover a signal with only $S$ active terms at unknown locations in any fixed basis scales linearly in $S$ and logarithmically in ambient-dimension $W$.

%
\section{Main Results: Sampling Architectures}
\label{sec:Samp-Archs} 

The main contribution of the paper is a design and theoretical analysis of a sampling architecture in Section \ref{sec:RD} that enables the sub-Nyquist acquisition of correlated signals. We state a Sampling Theorem \ref{thm:Exact-rec} that claims exact reconstruction of the signal ensemble using a much fewer samples compared to those dictated by Shannon-Nyquist sampling theorem. The proof of the theorem involves the construction of a dual certificate via golfing scheme to show that nuclear-norm minimization recovers the signal ensemble. Theorem \ref{thm:Exact-rec} is also of an independent interest as it is low-rank matrix recovery result form structured-random measurement ensemble.

 We begin with a straightforward architecture in Section~\ref{sec:arch-fixedproj} that minimizes the sample rate when the correlation structure is {\em known}.  We then combine our components from the last section in a specific way to create architectures that are provably effective under different assumptions on the signal ensemble.  The main sampling architecture in Section~\ref{sec:RD} uses random modulators prior to the ADCs; this architecture is effective when the energy in the ensemble is approximately uniformly dispersed across time.  Moreover, we expect the signal energy to be dispersed across array elements when the AVMM upfront does not mix the signals. In Section~\ref{sec:uniform-archs}, we present a variation of the above architecture in which ensembles are not required to be dispersed a priori instead the ensemble is preprocessed with LTI filters, and AVMM to ensure dispersion of energy across time, and array elements.

\subsection{Fixed projections for known correlation structure}
\label{sec:arch-fixedproj}

If the mixing matrix $\mA$ for ensemble $\mX_c(t)$ is known, then a straightforward way exists to sample the ensemble efficiently.  Let $\mA = \mU\mSigma\mV^*$ be the singular value decomposition of $\mA$, where $\mU$ is $M\times R$ matrix with orthogonal columns, $\mSigma$ is $R\times R$ diagonal matrix, and $\mV$ is $W\times R$ with orthogonal columns. An efficient way is to {\em whiten} ensemble $\mA$ with $\mU^*$ and sample the resulting $R$ signals (each at rate $W$). This scheme is shown in Figure \ref{fig:Fig4}.  $\mX$ can be written as a multiplication of matrix $\mU$ and $R \times W$ matrix $\mY$ containing the Nyquist samples of signals $\vx_1(t),\ldots,\vx_R(t)$ respectively in its $R$ rows. The discretized signal ensemble $\mX$ is then simply 
\[
\mX = \mU\mY.
\]
Knowing the correlation structure $\mU$, the ensemble $\mX$ and hence $\mX_c(t)$ (using sinc interpolation of samples in $\mX$) can be recovered using only the $RW$ samples in $\mY$. Observe that $RW$ is the optimal sampling rate as it only scales linearly with $R$, and not with $M$. 

\begin{figure}[htp]
	\begin{center}
		\includegraphics[trim = 2cm 12cm 2cm 0cm,scale = 0.6]{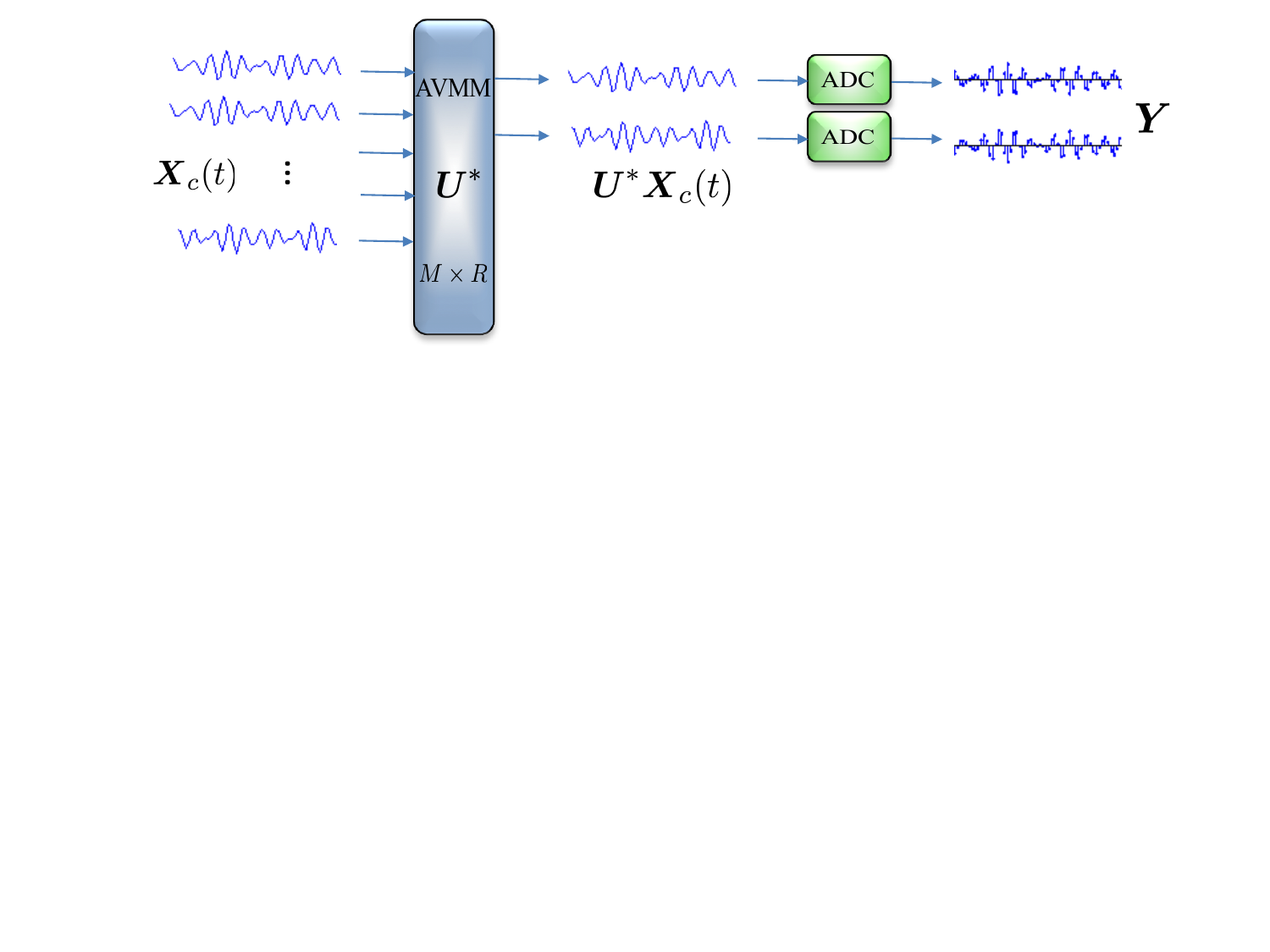}
	\end{center}
	\caption{\small\sl Known correlation structure $\mU$: Optimal sampling strategy is to {\em whiten} the ensemble $\mX_c(t)$ with $\mU^*$ and then sample and then sample each of the resultant $R$ signal at rate $W$.  Total $RW$ samples per second is optimal as it is the actual number of degrees of freedom in underlying $R$ independent signals each bandlimited to $W/2$.}
	\label{fig:Fig4}
\end{figure} 

In many interesting applications, the  correlation structure of the ensemble $\mX_c(t)$ is not known at the time of acquisition.  In this paper, we design sampling strategies that are blind to the correlation structure $\mU$ but are able achieve signal reconstruction at a near optimal sampling rate nonetheless by introducing AVMMs, filters, and modulators. Intuitively, the randomness introduced by these components disperses \textit{limited} information in the correlated ensemble across time and array elements; resultantly, the ADCs collect more generalized samples that in turn enable the reconstruction algorithm to operate successfully in the sub-Nyquist regime. 


\subsection{Architecture 1: Random sampling of time-dispersed correlated signals}
\label{sec:arch-randsample}

The architecture presented in this section, shown in Figure \ref{fig:Random-sampling}, consists of one non-uniform sampling (nus) ADC per channel.  Each ADC takes samples at randomly selected locations, and these locations are chosen independently from channel to channel.  Over the time interval $t \in [0,1)$, a nus ADC takes input signal $x_m(t)$ and returns the samples $\{x_m(t_k): t_k \in T_m \subset \{0,1/W, \ldots, 1-1/W\}$. The average sampling rate in each channel is $|T_m| = \O$.  Collectively, $M$ nus ADCs return $M\O$ random samples of the input signal ensemble on a uniform grid. 

\begin{figure}[ht]
	\begin{center}
		\includegraphics[trim=1.5cm 13cm 2cm 0.5cm,scale = 0.82]{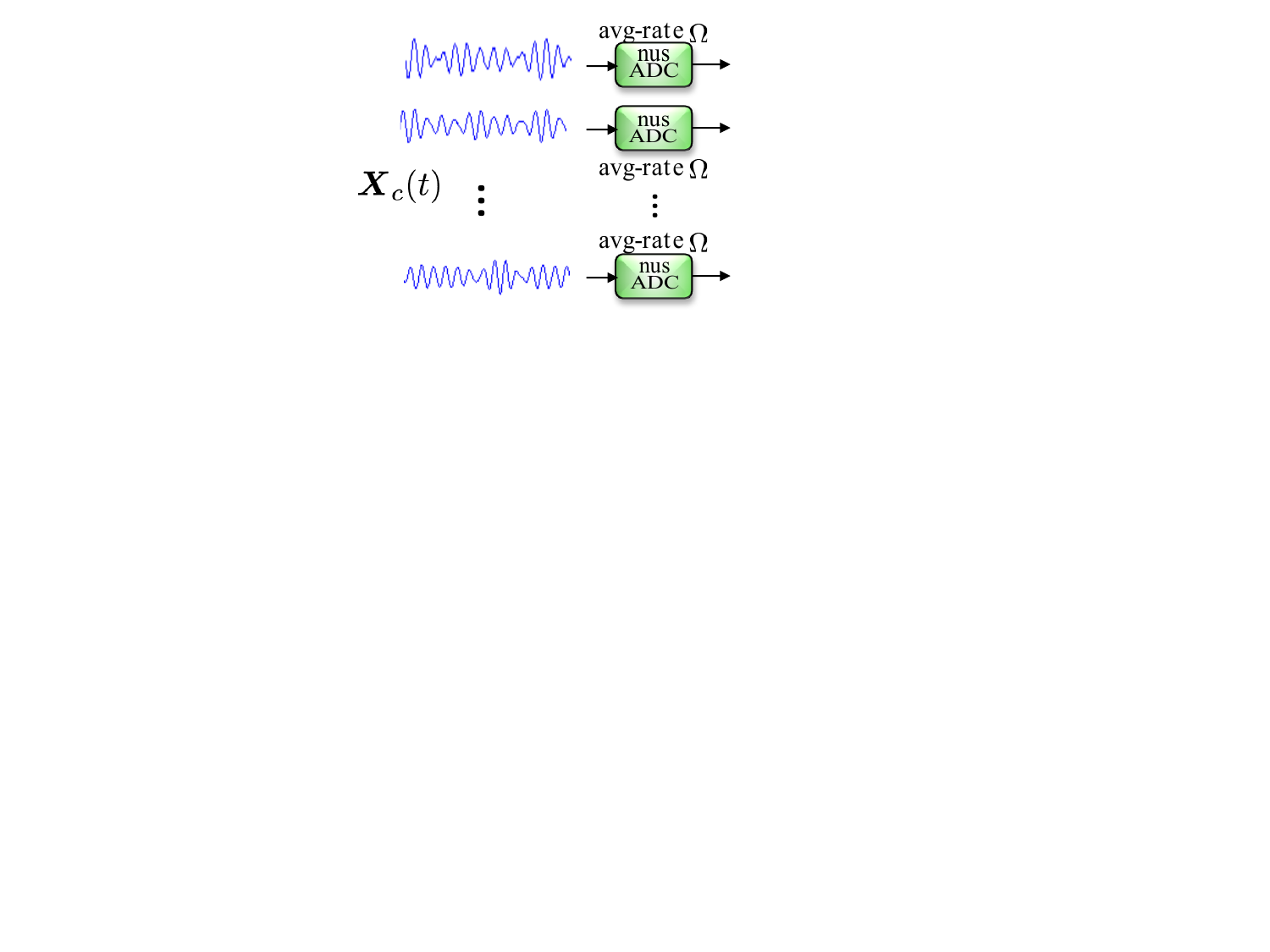}
	\end{center}
	\caption{\small\sl $M$ signals recorded by the sensors are sampled separately by the independent random sampling ADCs, each of which samples on a uniform grid at an average rate of $\O$ samples per second. This sampling scheme takes on the average a total of $M\O$ samples per second and is equivalent to observing $M\O$ entries of the matrix of samples $\mX$ in \eqref{eq:Cdef} at random.}
	\label{fig:Random-sampling}
\end{figure}

Sampling model is equivalent to observing $M\O$ randomly chosen entries of the matrix of samples $\mX$, defined in \eqref{eq:Cdef}. This problem is exactly the same as the matrix-completion problem \cite{candes09ex}, where given a few randomly chosen entries of a low-rank matrix enable us to {\em fill in} the missing entries under some incoherence assumptions on the matrix $\mX$. Since $\mX$ is rank-$R$, its SVD is 
\begin{equation}\label{eq:X-svd}
\mX = \mU\mSigma\mV^*,
\end{equation}
where $\mU \in \mathbb{R}^{M \times R}$, $\mSigma \in \mathbb{R}^{R \times R}$, and $\mV \in \mathbb{R}^{W \times R}$.  The coherence is then defined as 
\begin{equation}
\label{eq:coherence}
\kappa^2(\mU,\mV) := \max\left(\frac{M}{R}\max_{m \in [M]}\|\mU^* \ve_m\|_2^2,\frac{W}{R}\max_{\ell \in [W]}\|\mV^*\ve_\ell\|_2^2, \frac{MW}{R}\|\mU\mV^*\|_\infty^2\right). 
\end{equation}
For brevity, we will sometime drop the dependence on $\mU$, and $\mV$ in $\kappa^2(\mU,\mV)$. In the interest of readability, we assume without loss of generality here and in the rest of the write up that bandwidth $W$ of the signal is larger or at least equal to their number $M$, that is, $W\geq M$. Now the matrix-completion result \cite{recht11si} in the noiseless case asserts that if 
\begin{align}\label{eq:sampling-rate-Arch1}
\O \geq C\kappa^2(R/M)W\log^2W, 
\end{align}
then the solution of the nuclear-norm minimization in \eqref{eq:nuclearnorm_min} (with $\setA: \mathbb{R}^{M \times W} \rightarrow \mathbb{R}^{M\O}$ such that $\mathcal{A}$ maps $\mX$ to randomly chosen entries of $\mX$) exactly equals $\mX$ with high probability.  The result indicates that the sampling rate scales (within some $\log$ factors) with the number $R$ of independent signals rather than with the total number $M$ of signals in the ensemble. When the measurements $\vy$ are contaminated with additive measurement noise as in \eqref{eq:noisy-measurements} then the result in \cite{candes10ma} suggest that the solution $\hat{\mX}$ to a modified nuclear-norm minimization \eqref{eq:nuclearnorm_minnoisy} satisfies  
\[
\|\hat{\mX}-\mX\|^2_{\F} \leq C_{\kappa} M\delta^2,
\]
where $C_{\kappa}$ is a constant that depends on the coherence $\kappa$, defined in \eqref{eq:coherence}. 

As discussed before, the number of samples for matrix completion scale linearly with $\kappa^2$. The coherence parameter $\kappa^2$ quantifies the distribution of energy across the entries of $\mX$, and $\kappa^2$ is small for matrices with even distribution of energy among their entries; see, \cite{candes09ex} for details. In the signal reconstruction application under investigation here, this means that for successful recovery, a smaller sampling rate would suffice if the signals are well-dispersed across time and array elements. One can avoid this dispersion requirement by preprocessing the signals with AVMM, and filters. We will adopt this strategy in the construction of the main sampling architecture of this paper. 

\subsection{Architecture 2: The random modulator for correlated signals}
\label{sec:RD}

To efficiently acquire the correlated signal ensemble, the architecture shown in Figure \ref{fig:Rand-dem}, follows a two-step approach:

In the first step, the AVMM takes $M$ input to produce $N$ output signals, where $N/M= P >1$ meaning that the output signals are more than the inputs. For now, we take $N$ signals at the output to be just $P$ replicas of $M$ input signals without any mixing.\footnote{Our sampling theorem will show later that it suffices to take the ratio $P$ of the number of output to the input to be reasonably small. However, as will be suggested by our simulations, it seems $P = 1$ is always enough, and we believe $P>1$ merely a technical requirement arising due to the proof method.} This amounts to an $N \times M$ mixing matrix 
\begin{align}\label{eq:Mixing-matrix-Architecture-2}
\mA = \mI_M \otimes \tfrac{1}{\sqrt{P}}\mathbf{1}_P. 
\end{align}
The normalization by $P$ ensures that $\mA^*\mA = \mI_M$. We will take a more general random orthogonal $\mA$ in our next sampling architecture. 

In the second step, each of the $N$ output signals $\tilde{x}_1(t),\tilde{x}_2(t), \ldots, \tilde{x}_N(t)$ undergo analog preprocessing, which involves modulation, and low-pass filtering. The modulator takes an input signal $\tilde{x}_n(t)$ and multiplies it by a fixed and known $d_n(t)$. We will take $d_n(t)$ to be a binary $\pm 1$ waveform that is constant over an interval of length $1/W$. Intuitively, the modulation results in the diversification of the signal \textit{information} over the frequency band of width $W$. The diversified analog signals are then processed by an analog-low-pass filter; implemented using an integrator, see \cite{tropp2010beyond} for details. Each of the resultant signals is then acquired using $\Omega < W$ uniformly spaced samples per second. 

Our main sampling result in Theorem \ref{thm:Exact-rec} shows that exact signal reconstruction is achieved in the sub-Nyquist regime $\Omega < W$, in particular, we only roughly require $\Omega$ to be a factor $R/M \ll 1$ of the Nyquist rate $W$. Intuitively, the sub-Nyquist acquisition is possible as the signals are diversified across frequency using random demodulators, and therefore, every sample provides a generalized or global information.

\begin{figure}[ht]
	\begin{center}
		\includegraphics[trim=1.5cm 11cm 2cm 0.5cm,scale = 0.7]{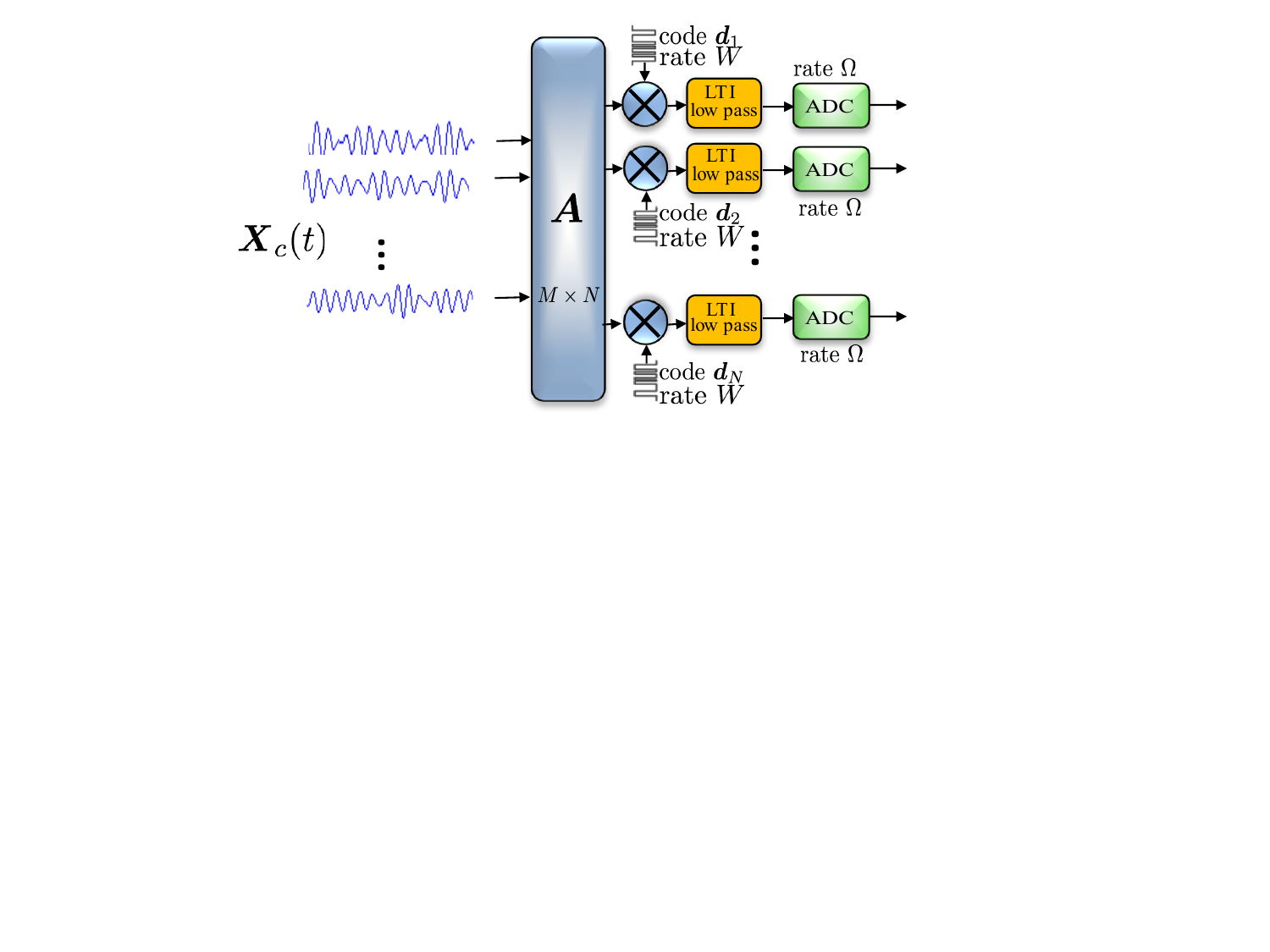}
	\end{center}
	\caption{\small\sl Architecture 2: Randomly modulated sampling: $M$ correlated signals in $\mX_c(t)$ are
		replicated $P$ times to produce $N$ output signals, this amounts to choosing $\mA = \mI_M \otimes\tfrac{1}{\sqrt{P}}\mathbf{1}_P \in \R^{N\times M}$ as the mixing matrix. In practice,  $P=1$ suffices. Signals are then preprocessed in analog using a bank of modulators, and low-pass filters. The resultant signal is then sampled uniformly by an ADC in each channel operating at rate $\Omega$ samples per second. The net sampling rate is $ \Omega N$ samples per second. }
	\label{fig:Rand-dem}
\end{figure}

\subsubsection{System Model}\label{sec:Sys-mat-form}
This section models\footnote{Some of the initial development in this section may resemble with \cite{tropp2010beyond}, but it is to be noted that compared to \cite{tropp2010beyond} the signal structure to be exploited here is \textit{correlations among the signals} and not the sparsity. This leads to a completely different development towards the end of this section.} the measured samples as the linear measurements of an unknown low-rank matrix. We will show that signal reconstruction in $t \in [0,1)$ from samples in the sub-Nyquist regime corresponds to recovering a $M \times W$ approximately rank-$R$ matrix from an under-determined set of linear equations. 

The input signal ensemble $\mX_c(t)$ is mixed using AVMM to produce an ensemble of $N$ signals $\mA\mX_c(t)$. Let us denote the individual $N$ signals at the output of AVMM by $\tilde{x}_1(t),\tilde{x}_2(t),\ldots,\tilde{x}_N(t).$  Since mixing is a linear operation, every signal in the ensemble $\mA\mX_c(t)$ is bandlimited just as was the case with $\mX_c(t)$ in \eqref{eq:Cdef}; therefore, the DFT coefficients of the mixed signals are simply 
\begin{equation}\label{eq:Cz-def}
\widetilde{\mC} = \mA\mC.
\end{equation}
Each signal $\tilde{x}_n(t)$ at the output of AVMM is then multiplied by a corresponding binary sequence $d_n(t)$ alternating at rate $W$. Each of the binary sequences $d_1(t), d_2(t), \ldots, d_N(t)$ will be generated randomly, and independently.  The output after modulation in the $n$th channel is 
\[
y_n(t) = \tilde{x}_n(t)\cdot d_n(t), \quad n \in [N], \mbox{ and } t \in [0,1).
\]
The modulated outputs $y_n(t)$ are then low-pass filtered using an integrator, which integrates $y_n(t)$ over an interval of width $1/\O$ and the result is then sampled at rate $\O$ using an ADC. The $\ell$th sample acquired by the ADC in the $n$th channel is
\[
y_{n}[\ell] = \int_{(\ell-1)/\O}^{\ell/\O} y_n(t) dt, \quad \ell \in [\O].
\]
The integration operation commutes with the modulation process; hence, we can equivalently integrate the signals $z_n(t)$ over the interval of width $1/W$, and treat them as samples $\mZ_0 \in \R^{M \times W}$ of the ensemble $\mZ_c(t)$. The entries $Z_0[n,\ell]$ of the matrix $\mZ_0$ are 
\begin{align} \label{eq:Z0-def}
Z_0[n,\ell] &= \int_{(\ell-1)/W}^{\ell/W} \tilde{x}_n(t) dt =  \sum_{\o = -B}^B \widetilde{C}[n,\o] \left[\frac{\e^{\j 2\pi \o / W} - 1}{\j 2\pi \o}\right]e^{-\iota 2\pi\o \ell/W},
\end{align}
where $\widetilde{C}[n,\omega]$ are the entries of the matrix $\widetilde{\mC}$ defined in \eqref{eq:Cz-def}, and the bracketed term representing the low-pass filter 
\[
L[\o] = \left[\frac{\e^{\j 2\pi \o / W} - 1}{\j 2\pi \o}\right], \quad \o = -B,-B+1,\ldots,B,
\]
where $W=2B+1$, as defined in \eqref{eq:Signaltx}.  We will denote by $\mL$ as a diagonal matrix containing $L[\o]$ along the diagonal. It is important to note that $\mL$ is invertible as $L[\o]$ does not vanish on any $\o = -B,-B+1,\ldots,B$. In view of \eqref{eq:Z0-def}, it is clear that
\begin{equation}
\mZ_0 = \widetilde{\mC}\mL\mF^* = \mA\mC\mL\mF^* = \mA\mX_0\label{eq:X0def}
\end{equation}
where $\mX_0 = \mC\mL\mF^*$ inherits its low-rank structure from $\mX_c(t)$.

Since we have already carried out integration over intervals of length $1/W$, the action of modulator followed by integration over $1/\O$ now simply reduces to randomly, and independently flipping every entry of $\mZ_0$ and adding consecutive $W/\O$ such entries in a given row to produce the value of the sample acquired by the ADC. Mathematically, we can write this concisely by defining a vector $\vd_{n\ell}$ supported on an index set 
\begin{align}\label{eq:setBl}
\setB_\ell  := \{ (\ell-1)W/\O+1:\ell W/\O\},\quad \ell \in [\Omega],
\end{align}
of size $|\setB_\ell| = W/\O$, where we are assuming\footnote{a slight modification of this can result in an argument when $\O$ is not a factor of $W$; for details, see \cite{tropp2010beyond}} for simplicity that $\O$ is a factor of $W$. On the support set $\setB_\ell$, the entries of the vector $\vd_{n\ell}$ are independent binary $\pm 1$ random variables, and are zeros on $\setB_\ell^c$. Moreover, assume that $\va_n^*$ are the rows of  $\mA$. With these notations in place, we can concisely write the $\ell$th sample in $t \in [0,1)$ in the $n$th branch as 
\begin{equation}\label{eq:meas}
y_{n}[\ell]: = \va_n^*\mX_0\vd_{n\ell}, ~ (n,\ell) \in [N]\times [\O]. 
\end{equation}
All this shows is that the samples taken by the ADC in the sampling architecture in Figure \ref{fig:Rand-dem} are linear measurements of an underlying low-rank matrix $\mX_0 \in \R^{M \times W}$, defined in \eqref{eq:X0def}. The rank of $\mX_0$ does not exceed $R$ --- recalling from Section \ref{sec:signal-model} that $R$ constitutes the number of linearly independent signals in the ensemble $\mX_c(t)$. Our objective is to recover $\mX_0$ from a a few linear measurements $y_n[\ell]$, which amounts to reconstructing $\mX_c(t)$ at a sub-Nyquist rate sampling. 
\subsection{Matrix recovery}
Define a linear map $\setA: \mX_0 \rightarrow \vy$, where $\vy$ is a length $N \O$ vector containing linear measurements $y_n[\ell]$ in \eqref{eq:meas} as its entries. Formally, 
\begin{equation}\label{eq:cA-def}
\setA(\mX_0) = \{ \va_n^*\mX_0\vd_{n\ell} : (n,\ell) \in [N]\times[\Omega]\}.
\end{equation}
We are mainly interested in the scenario when the linear map $\setA$ is under determined, that is, the number of measurements $N\O$ is much smaller than the number of unknowns $MW$. Therefore, to uniquely determine the true solution $\mX_0$, we solve a nuclear-norm penalized optimization program: 
\begin{align}
&\underset{\mX}{\text{argmin}}~\left\|\mX\right\|_*\label{eq:nuclearnorm_min}~\mbox{subject to}~\vy = \setA(\mX),
\end{align}
where $\left\|\mX\right\|_*$ is the nuclear norm (the sum of the singular values of $\mX$).  The nuclear norm penalty encourages the solution to be low rank \cite{fazel02ma}, and has concrete performance guarantees when the linear  map $\setA$ obeys certain properties \cite{recht10gu}.  In case of noisy measurements
\begin{equation}\label{eq:noisy-measurements}
\vy = \setA(\mX_0) + \vxi
\end{equation}
with bounded noise $\|\vxi\|_2 \leq \delta$, we solve the following quadratically constrained convex optimization program 
\begin{align}\label{eq:nuclearnorm_minnoisy}
&\underset{\mX}{\text{argmin}}~\left\|\mX\right\|_* ~\mbox{subject to}~\|\vy - \setA(\mX)\|_2 \leq \delta. 
\end{align}
This optimization program is also provably effective; see, for example, \cite{candes10ma,mohan10ne} for suitable $\setA$.

\subsubsection{Sampling theorem: Exact and stable recovery}\label{sec:Samp-thm}

The unknown matrix $\mX_0$ in \eqref{eq:X0def} is at most rank-$R$, and assume $\mX_0 = \mU\mSigma\mV^*$ is  its reduced form SVD, where $\mU\in \R^{M \times R}$, and $\mV \in \R^{ W \times R}$ are the  matrices of left and right singular vectors, respectively; and $\mSigma \in \R^{R \times R}$ is a diagonal matrix of singular values. Define coherences of $\mX_0$ as 
\begin{align}
\label{eq:coherences-random-demodulator}
&\mu^2(\mU) := \frac{M}{R}\max_{m \in [M]} \|\mU^*\ve_m\|_2^2, \quad \rho^2(\mV) : = \frac{\O}{R}\max_{\ell \in [\O]}\|\mV^*\mI_{\setB_\ell}\|_{\F}^2, \quad \text{and} \quad\notag\\
& \qquad\qquad \qquad \nu^2(\mU,\mV) := \frac{M\Omega}{R} \max_{m \in [M]} \max_{\ell \in [\Omega]} \|\ve^*_m\mU\mV^*\mI_{\setB_\ell}\|_2^2,
\end{align}
where $\mI_{\setB_\l}$ is a diagonal $W \times W$ matrix containing ones at the diagonal positions indexed by $\setB_\l$. We may sometime just work with notations $\mu^2, \rho^2$, and $\nu^2$ and drop the dependence on $\mU$, and $\mV$  when it is clear from the context.  It can easily be verified that $ 1 \leq \mu^2 \leq M/R$. In a similar manner, one can show that  $1 \leq \rho^2 \leq W/R$. To see this, notice that 
\begin{align*}
\rho^2 \geq \frac{\O}{R} \|\mV^*\mI_{\setB_\ell}\|_{\F}^2 \implies \O \rho^2 \geq \frac{\O}{R} \sum_{\ell=1}^\O \|\mV^*\mI_{\setB_\l} \|_\F^2 = \frac{\O}{R} \cdot R,
\end{align*}
that is, $\rho^2 \geq 1$. Using the fact that $\|\mV^*\mI_{\setB_\l}\|_{\F} \leq \|\mV\|^2 \|\mI_{\setB_\l}\|_\F^2 = W/\O$, the upper bound also follows. Finally,  similar techniques also show that $ 1 \leq \nu^2 \leq M\Omega/R$. One can attach meaning to the values of coherences in the context of sampling application under consideration. For example, the smallest value of $\rho^2$ is achieved the energy of $\mX_0$ is roughly equally distributed among the columns indexed by $\setB_1,\setB_2,\ldots,\setB_\O$. In the context of the sampling problem, this means that the energy in the signal ensemble $\mX_c(t)$ should be dispersed equally across time. Similarly, the coherence  $\mu^2$ quantifies the spread of signal energy across array elements, and $\nu^2$ measures the dispersion of energy across both the time and array elements. Let us define 
\begin{align}\label{eq:max-coherence-def}
\vartheta^2(\mU,\mV) : = \max(\mu^2(\mU),\rho^2(\mV), \nu^2(\mU,\mV)).
\end{align}
We are now ready to state our main result that dictates the minimum sampling rate $\O$ at which each ADCs needs to be operated to guarantee the reconstruction of signal ensemble $\mX_c(t)$. 
\begin{thm}\label{thm:Exact-rec}
	  Correlated signal ensemble $\mX_c(t)$ in \eqref{eq:lowrankensemble} can be acquired using the sampling architecture in Figure \ref{fig:Rand-dem} by operating each of the ADCs at a rate 
	\begin{equation}\label{eq:Omega-bound}
	\O \geq C_\beta\vartheta^2 \frac{R}{M} W\log^2 W,
	\end{equation}
	where $C_\beta$ is a universal constant that only depends on the fixed parameter $\beta \geq 1$. In addition, the ratio of the number of output to the input signals in AVMM must satisfy $N/M \geq C\log W$, where $C$ is a numerical constant. The exact signal reconstruction can be achieved with probability at least $1-\mathcal{O}(W^{1-\beta})$ by solving the nuclear-norm minimization program in  \eqref{eq:nuclearnorm_min}. 
\end{thm}

The result indicates that $M$ well spread out ($\vartheta^2 \approx 1$) correlated signals can be acquired by operating each ADC in Figure \ref{fig:Rand-dem} at a rate of $R/M \ll 1 $ times the Nyquist-rate $W$ to within $\log$ factors. Moreover, we also require the number $N$ of output signals at the AVMM to be larger than number $M$ of input signals by a $\log$ factor. However, we believe this is merely an artifact of the proof technique and our experiments also corroborate that successful recovery is always obtained for $\Omega$ satisfying \eqref{eq:Omega-bound} even when $N = M$ or $P=1$. Also note that the result in Theorem \ref{thm:Exact-rec} assumes without loss of generality that $W\geq M$. In the other case, when  $M \geq W$,  the sufficient sampling rate at each ACC can be obtained by replacing $W$ in \eqref{eq:Omega-bound} with $M$. 

Another important observation is that the sampling rate scales linearly with coherence $\vartheta^2$ implying that the sampling architecture is not as effective for correlated signals concentrated across time. To remedy this shortcoming, a preprocessing step using random filters, a mixing AVMM can be added to ensure signals are well-dispersed across time, and array elements. 

\subsubsection{Stable recovery}

In a realistic scenario, the measurements \eqref{eq:meas} are almost always contaminated with noise $\xi_n[\ell]$
\begin{align}\label{eq:noisy-mesaurements-entrywise}
y_n[\ell] = \va_n^*\mX_0\vd_{n\ell} + \xi_n[\ell], \quad (n,\ell) \in [N] \times [\O]
\end{align}
compactly expressed using the vector equality in \eqref{eq:noisy-measurements}. 
In the case, when the noise is bounded, i.e., $\sum_{n,\ell} |\xi_{n}[\ell]|^2 \leq \delta^2$, then following the template of the proof in \cite{candes10ma}, it can be shown that under the conditions of Theorem \ref{thm:Exact-rec}, the solution $\widehat{\mX}$ of \eqref{eq:nuclearnorm_minnoisy} obeys 
\begin{equation}\label{eq:st-rec-Y.Plan}
\|\widehat{\mX}-\mX_0\|_{\F} \leq C \sqrt{M} \delta
\end{equation}
with high probability; for more details on this, see a similar stability result in Theorem 2 in \cite{ahmed2012blind}. The upper bound above is suboptimal by a factor of $\sqrt{\min(W,M)}$. In theory, we can improve this suboptimal result and show the effectiveness of the nuclear norm penalty by analyzing a different estimator:
\begin{align}\label{eq:KLT-estimator}
\underset{\mX}{\text{argmin}}~\|\mX-\setA^*(\vy)\|_{\F}^2 + \lambda \|\mX\|_*.
\end{align}
This estimator was proposed in \cite{koltchinskii10nu}, and can be theoretically shown to obey essentially optimal stable recovery results. Using the fact that $\widehat{\mX}$ is the minimizer of \eqref{eq:KLT-estimator} if and only if  $\mathbf{0} \in \partial\big(\|\mX-\setA^*(\vy)\|_{\F}^2 + \lambda \|\mX\|_*\big)$,  one can show\cite{koltchinskii10nu} that the estimate $\widehat{\mX}$ is a simple soft thresholding of the singular values of the matrix  $\setA^*(\vy) \in \R^{M \times W}$
\begin{align}\label{eq:KLT-solution}
\widehat{\mX} = \sum_{r}\bigg(\sigma_r(\setA^*(\vy))-\frac{\lambda}{2}\bigg)_+\vu_r(\setA^*(\vy))\vv_r^*(\setA^*(\vy)),
\end{align}
where $x_+ = \max(x,0)$; in addition, $\vu_r(\setA^*(\vy))$, and $\vv_r(\setA^*(\vy))$ are the left and right singular vectors of the matrix $\setA^*(\vy)$, respectively; and $\sigma_r(\setA^*(\vy))$ is the corresponding singular value. 

In comparison to the estimator \eqref{eq:KLT-estimator}, the matrix Lasso in \eqref{eq:nuclearnorm_minnoisy} does not use the knowledge of the known distribution of $\setA$ and instead minimizes the empirical risk $\|\vy-\setA(\mX)\|_2 = \|\vy\|_2^2 - 2\<\vy,\setA(\mX)\> + \|\setA(\mX)\|_2^2$. Knowing the distribution, and the fact that $\E \setA^*\setA = \mathcal{I}$ holds in our case, we replace $\|\setA(\mX)\|_2^2$, by its expected value $\E \|\setA(\mX)\|_2^2 = \|\mX\|_{\F}^2$ in the empirical risk to obtain the estimator in \eqref{eq:KLT-estimator} by completing the square. Although the KLT estimator is easier to analyze, and will be shown to give optimal stable recovery results in theory, but it does not empirically perform as well as matrix Lasso in \eqref{eq:nuclearnorm_minnoisy}.

We quantify the strength of the noise vector $\vxi \in \R^{N\Omega}$ through its Orlicz-2 norm. For a random vector $\vz$, we define 
\[
\|\vz\|_{\psi_2} := \inf \{ u > 0: \E \big[ e^{\|\vz\|_2^2/u^2} \leq 2\big]\},
\]
and for scaler random variables, we simply take $\vz \in \R^1$ in the above definition. The Orlicz-2 norm is finite if the entries of $\vz$ are subgaussain, and is proportional to variance, if the entries are Gaussian. We assume that the entries of the noise vector $\vxi$ obey
\begin{align}\label{eq:noise-statistics}
\|\xi_n[\ell]\|_{\psi_2} \leq c\delta\tfrac{1}{N\Omega}, ~ \mbox{and}~ \|\vxi\|_{\psi_2} \leq c\delta.
\end{align}
 With this the following result is in order.
\begin{thm}\label{thm:stable-rec}
	 Fix $\beta \geq 1$. Given measurements $\vy$ of $\mX_0$ in \eqref{eq:noisy-mesaurements-entrywise} contaminated with additive noise $\vxi$ with statistics in \eqref{eq:noise-statistics}, the solution $\widehat{\mX}$ to \eqref{eq:KLT-estimator} obeys 
	\begin{equation}\label{eq:st-rec}
	\|\widehat{\mX} -\mX_0\|^2_{\F} \leq C\max(\delta,\delta^2)
	\end{equation}
	with probability at least $1-\setO(W^{-\beta})$ whenever $\Omega \geq C_\beta \nu^2 (R/M)W\log^{3/2}W$, where $C_\beta$ is a universal constant depending only on $\beta$. 
\end{thm}
Roughly speaking, the stable recovery theorem states that the nuclear norm penalized estimators are stable in the presence of additive measurement noise. The results in Theorem \ref{thm:stable-rec} are derived assuming that $\xi_{n}[\ell]$ are random with statistics in \eqref{eq:noise-statistics}. In contrast, the stable recovery results in the compressed sensing literature only assume that the noise is bounded, i.e., $\|\vxi\|_2 \leq \delta$, where $\vxi$ is the noise vector introduced earlier. Here, we give a brief comparison of Theorem \eqref{thm:stable-rec} with the stable recovery results in \cite{candes10ma,fazel2008compressed }. 

Compare the result in \eqref{eq:st-rec-Y.Plan} with \eqref{eq:st-rec}, it follows that our results improve upon the results in \cite{candes10ma} by a factor of $1/\sqrt{M}$. We will also compare our stable recovery results against the stable recovery results derived in \cite{fazel2008compressed}. The result roughly states if the linear operator $\setA$ satisfies the matrix RIP \cite{recht10gu}, and $\|\vxi\|_2 \leq \delta$, then the solution $\widehat{\mX}$ to \eqref{eq:nuclearnorm_minnoisy} obeys 
\begin{equation}\label{eq:st-rec-Fazel}
\|\widehat{\mX}-\mX_0\|_{\F} \leq C \delta.
\end{equation}
The above result is essentially optimal stable recovery result. In comparison to \eqref{eq:st-rec-Fazel}, the result in \eqref{eq:st-rec} is also optimal, however, we prove it for a different estimator and under a statistical bound on the noise term  $\|\vxi\|_{\psi_2} \leq \delta$. In addition, we also donot require the matrix RIP for $\setA$, which is generally required to prove optimal results of the form of \eqref{eq:st-rec-Fazel}.

\subsection{Architecture 3: Uniform sampling architecture}
\label{sec:uniform-archs}
The discussion in Section \ref{sec:arch-randsample}, and  the result in Theorem \ref{thm:Exact-rec}  suggest that sampling rate sufficient for exact recovery using the architecture 1 and 2 scales linearly with the coherence parameter $\mu^2$, and $\vartheta^2$, respectively.  As discussed earlier, the coherence parameters quantify the energy dispersion in the correlated signal ensemble $\mX_c(t)$ across time and array elements. Ideally, we would like the sampling rate $\Omega$ to only scale with factor of $W$, and be independent of signal characteristics (coherences). To achieve this, signals are preprocessed with random filters and AVMM so that signal energy is evenly distributed across time, and array elements. The resultant signals are the randomly modulated, low-pass filtered, and sampled uniformly at a rate $\Omega$. The modified sampling architectures are depicted in Figure \ref{fig:Uniform-Rand-Samp} and \ref{fig:Uniform-RD}. 

\begin{figure}[ht]
	\begin{center}
		\includegraphics[trim= 1cm 12.5cm 2.5cm 0cm,scale = 0.68]{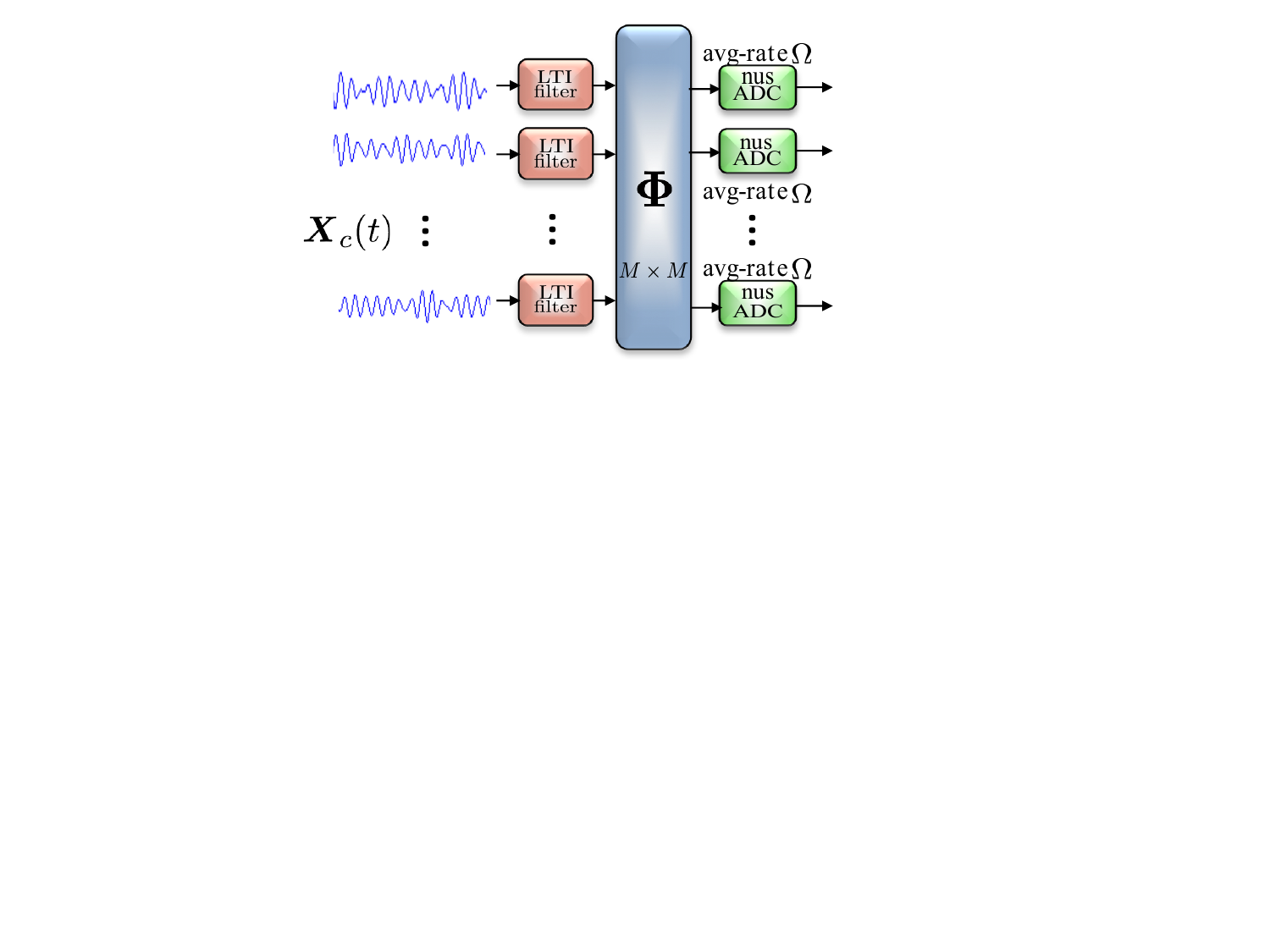}
	\end{center}
	\caption{\small\sl Architecture 3: Analog vector-matrix multiplier (AVMM) takes random linear combinations of $M$ input signals to produce $M$ output signals. This equalizes energy across channels. The random LTI filters convolve the signals with a diverse waveform that results in dispersion of signals across time. The resultant signals are then sampled, at locations selected randomly on a uniform grid, at an average rate $\O$, using a non-uniform sampling (nus) ADC in each channel.}
	\label{fig:Uniform-Rand-Samp}
\end{figure}

\begin{figure}[ht]
	\begin{center}
		\includegraphics[trim= 4cm 9.5cm 2.5cm 0cm,scale = 0.72]{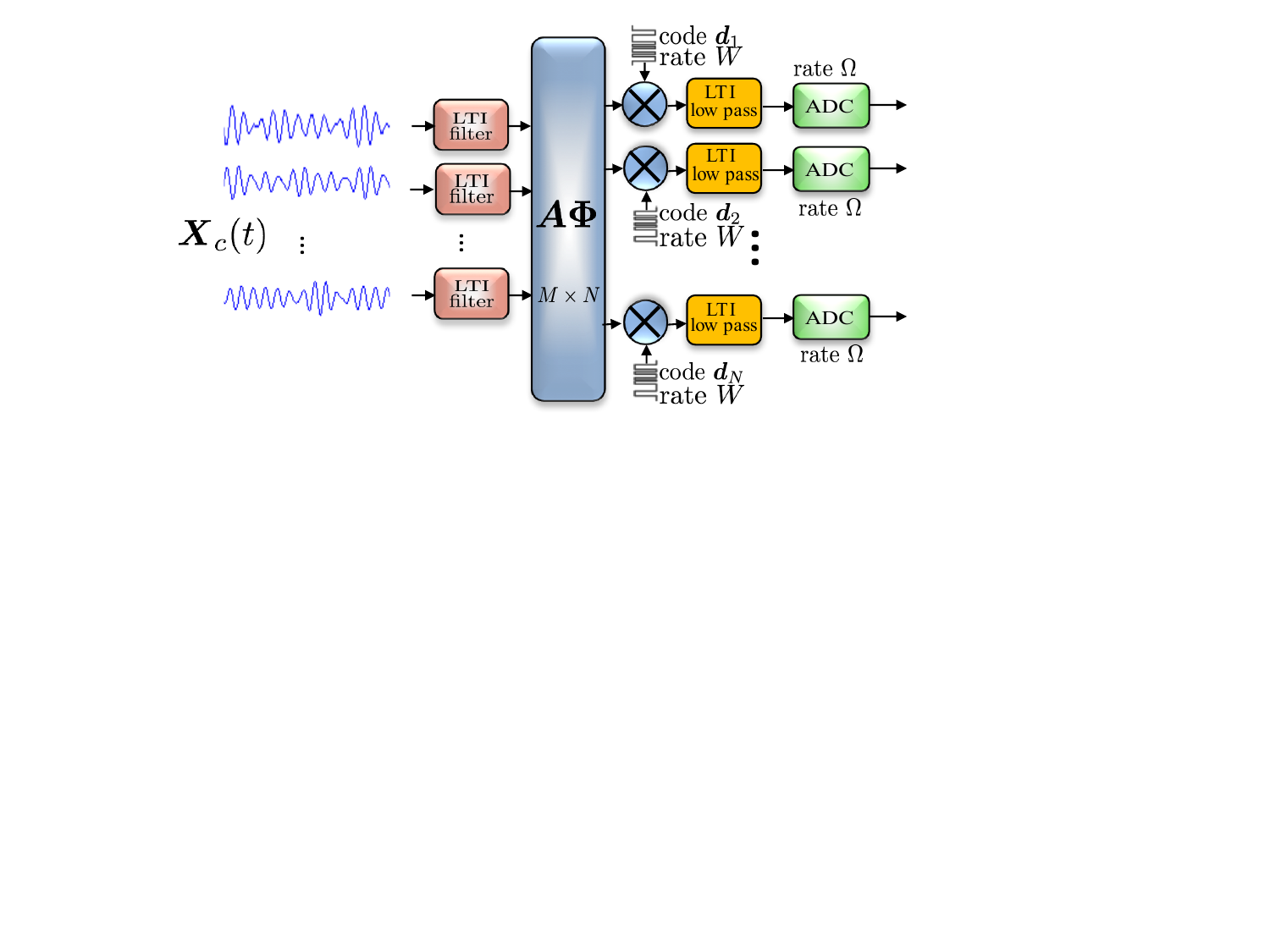}
	\end{center}
	\caption{\small\sl Architecture 4: Random LTI filters disperse each of the $M$ signal across time. An Analog vector-matrix multiplier (AVMM) takes random linear combinations of $M$ input signals to produce $N$ output signals. This amounts to choosing $[\mPhi \otimes \tfrac{1}{\sqrt{P}}\mathbf{1}_P] = \mA\mPhi$ as the mixing matrix, where $\mA$ is as in \eqref{eq:Mixing-matrix-Architecture-2}, and $\mPhi$ is an $M \times M$ dense, random-orthogonal matrix. The well dispersed signals across time and array elements are now randomly modulated, low-pass filtered, and sampled at rate $\Omega$.}
	\label{fig:Uniform-RD}
\end{figure}

Recall that random LTI filters are all pass, and convolve the signals with a diverse impulse response $h_c(t)$, which disperses signal energy over time w.h.p. (see Lemma \ref{lem:coherence}). We will use the same random LTI filter $h_c(t)$ in each channel. The action of the random convolution\cite{romberg2009compressive} of $h_c(t)$ with each signal in the ensemble can be modeled by the right multiplication of a circulant random orthogonal matrix $\mH \in \R^{W \times W}$  with the underlying low-rank $\mX_0$ in \eqref{eq:X0def}. The AVMM takes the random linear combination of $M$ input signals to produce $N$ output signals, which then equalizes w.h.p., the signal energy across array elements regardless of the initial energy distribution.  

As discussed earlier that the action of AVMM is left multiplication of $\mA$ with the low-rank ensemble $\mX_c(t)$. In Architecture 3, the AVMM is $M$-in-$M$-out, and to ensure mixing of signals across array elements, we take the mixing matrix to be an $M \times M$ random orthonormal matrix $\mPhi$. Thus, the samples collected in Architecture 3 are not the subset of the entries of $\mX$, defined in \eqref{eq:Cdef} but of 
\begin{align}\label{eq:Xtilde-def}
\tilde{\mX} = \mPhi\mX\mH.
\end{align} 

In Architecture 4, the AVMM is modified from $\mA= [\mI_M \otimes \tfrac{1}{\sqrt{P}}\mathbf{1}_P] \in \R^{N \times M}$ in \eqref{eq:Mixing-matrix-Architecture-2} to 
\begin{align}\label{eq:Mixing-matrix-Architecture-4}
[\mPhi \otimes \tfrac{1}{\sqrt{P}}\mathbf{1}_P] = \mA\mPhi,
\end{align}
 where $\mPhi \in \R^{M\times M}$ is a random orthonormal matrix. This implies that unlike the samples $\vy = \setA(\mX_0)$ , the Architecture 4  collects $\vy = \setA(\tilde{\mX}_0)$, where 
 \begin{align}\label{eq:X0tilde-def}
 \tilde{\mX}_0 = \mPhi\mX_0\mH,
 \end{align}
 and $\setA$ is same as defined in \eqref{eq:cA-def}.

Both \eqref{eq:Xtilde-def}, and \eqref{eq:X0tilde-def} multiply the matrix of samples $\mX$, and $\mX_0$ with random orthogonal matrices on the left, and right. This multiplication results in modifying the singular vectors  $\mU \in \R^{M\times R}$, and $\mV \in \R^{W \times R}$ of the matrix of samples (either $\mX$, or $\mX_0$) to $\widetilde{\mU} = \mPhi\mU \in \R^{M \times R}$ and $\widetilde{\mV} = \mH\mV \in \R^{W \times R}$. Note that matrix $\tilde{\mX}$, and $\tilde{\mX}_0$ are an isometry with, and have the same rank as $\mX$, and $\mX_0$, respectively. The new left and right singular vectors $\tilde{\mU}$ and $\widetilde{\mV}$ of $\tilde{\mX}$, or $\tilde{\mX}_0$ are in some sense random orthogonal matrices and hence, incoherent w.h.p. 

The following lemma shows the incoherence of matrix $\tilde{\mU}$, and $\tilde{\mV}$. 
\begin{lem}\label{lem:coherence}
	Fix matrices $\mU \in \mathbb{R}^{M \times R}$ and $\mV \in \mathbb{R}^{W \times R}$ of the left and right singular vectors, respectively. Create random orthonormal matrices  $\mPhi \in \R^{M\times M}$ and $\mH \in \R^{W \times W}$. Let $\widetilde{\mU} = \mPhi\mU$, and $\widetilde{\mV} = \mH\mV$, and the coherences $\kappa^2(\widetilde{\mU},\widetilde{\mV})$, and $\vartheta^2(\widetilde{\mU},\widetilde{\mV})$ be as defined in \eqref{eq:coherence}, \eqref{eq:max-coherence-def}. Then for a $\beta \geq 1$, the  following conclusions
	 \begin{equation}\label{eq:kappa-bound}
	 \kappa^2(\widetilde{\mU},\widetilde{\mV}) \leq C_\beta\log W\max{(1,\frac{1}{R}\log M)},
	 \end{equation}
	and 
	 \begin{equation}\label{eq:vartheta-bound}
	 \vartheta^2(\widetilde{\mU},\widetilde{\mV}) \leq C_\beta\log W\max{(1,\frac{1}{R}\log M)}
	 \end{equation}
	  each holding with probability exceeding $1-\mathcal{O}(W^{-\beta})$.
  \end{lem}
Proof of Lemma~\ref{lem:coherence} is presented in Section ~\ref{sec:coherence-proof}. 
In light of \eqref{eq:Xtilde-def}, it is clear that samples collected using Architecture 3 are randomly selected subset of the entries of $\tilde{\mX}$, and using the result in \eqref{eq:kappa-bound}, and \eqref{eq:sampling-rate-Arch1}, the sufficient sampling rate for the successful reconstruction of signals becomes 
\[
\Omega \geq C\frac{1}{M}\max(R,\log M)W\log^3 W.
\]
In light of \eqref{eq:X0tilde-def}, it is clear that the samples collected using Architecture 4 are the same as in \eqref{eq:meas} with $\mX_0$ replaced by $\widetilde{\mX}_0$. 
With this observation, combining the bound on $\nu^2$ in Lemma ~\ref{lem:coherence} with Theorem \ref{thm:Exact-rec} immediately provides with the following corollary that dictates the sampling rate sufficient for exact recovery using the uniform sampling architecture in Figure \ref{fig:Uniform-RD}.
\begin{cor}\label{cor:Exact-rec}
	Fix $\beta \geq 1$. The correlated ensemble $\mX_c(t)$ in \eqref{eq:lowrankensemble} can be exactly reconstructed using the optimization program in \eqref{eq:nuclearnorm_min}  with probability at least $1-\mathcal{O}(W^{-\beta})$ from the samples collected by each of the ADC in Figure \ref{fig:Uniform-RD} at a sub-Nyquist rate:
	\begin{equation}\label{eq:Omega-bound2}
	\O \geq C_\beta \frac{1}{M}\max(R,\log M)W\log^4W,
	\end{equation}
	where $C_\beta$ is a universal constant depending only on $\beta$. In addition, the ratio of the number of output to the input signals in AVMM must satisfy $N/M \geq C\log W$ for a sufficiently large constant $C$. 
\end{cor}

%
%
\section{Numerical Experiments}\label{sec:Exps}
In this section, we study the performance of the proposed sampling architectures with some numerical experiments. We mainly show that a correlated ensemble $\mX_c(t)$ in \eqref{eq:lowrankensemble} can be acquired by only paying a small factor on top of the optimal sampling rate of roughly $RW$. We then show the distributed nature of the sampling architecture in Figure \ref{fig:Rand-dem} by showing that increasing the number of ADCs (or the array elements), the sampling burden on each of the ADC can be reduced as the net sampling rate is shared evenly among the ADCs. Finally, we show that the reconstruction algorithm is robust to additive noise.

\subsection{Sampling performance}
In all of the experiments in this section, we generate the unknown rank-$R$ matrix $\mX_0$ synthetically  by multiplying tall $M \times R$, and fat $R \times W$ Gaussian matrices. Our objective is to recover a batch of $M  = 100$ signals, with $W = 1024$ samples taken in a given window of time using the sampling architecture in Figure \ref{fig:Rand-dem}. We take $P = 1$ or $N=M$ for all these experiments, and the results hint that $P>1$ or $N>M$ in Theorem \ref{thm:Exact-rec} is only a technical requirement due to the proof technique.  We will use the following parameters to evaluate the performance of the sampling architecture:
\[
\mbox{Oversampling factor}: \eta = \frac{M\O}{R(W+M-R)},
\]
where the oversampling factor is the ratio between the cumulative sampling rate, and the inherent unknowns in $\mX_0$. The successful reconstruction is declared when the relative error obeys 
\[
\mbox{Relative error}: = \frac{\|\widehat{\mX}-\mX_0\|_{\F}}{\|\mX_0\|_{\F}} \leq 10^{-2}.
\]

The first experiment shows a graph, in Figure \ref{fig:RDLvsR}, between $\eta$, and $R$. Each point, marked with a black dot, represents the minimum sampling rate required for the successful reconstruction of an $\mX_c(t)$ with a specific $R$. The probability of success for each point is $0.99$, and is computed empirically by averaging over 100 independent iterations. The blue line shows the least-squares fit of the black dots. It is clear from the plot that the for reasonably large values of $R$, the sampling rate is within a small constant of the optimal rate $R(W+M-R)$. 

In context of the application, and under the narrow-band assumption described in Section \ref{sec:APP},  the graph in Figure \ref{fig:RDLvsM} shows that for a fixed number of sources $R = 10$, the sufficient sampling rate $\O$ is inversely proportional to number $M$ of the receiver array elements. Each black dot represents the minimum sampling rate required for the successful reconstruction with probability $0.99$. The blue line is the least-squares fit of these marked points. In other words, Figure \ref{fig:RDLvsM} illustrates the relationship between the number $M$ of ADCs and the sampling rate $\O$ for a fixed number of sources $R = 10$. Importantly, an increase in the receiver array elements reduces the sampling burden on each of the ADCs. 

\begin{figure}[ht]
	\centering
	\subfigure[]{
		\includegraphics[trim=2.5cm 7.5cm 2.5cm 7.5cm,scale = 0.45]{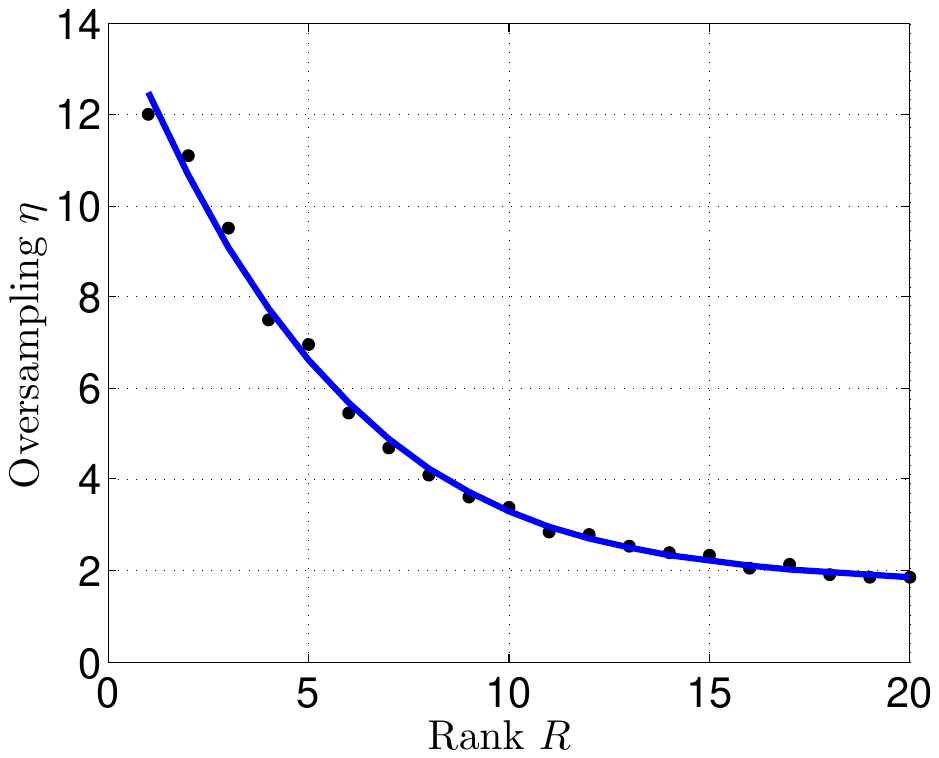}
		\label{fig:RDLvsR}}
	\subfigure[]{
		\includegraphics[trim=2.5cm 7.5cm 2.5cm 7.5cm,scale = 0.45]{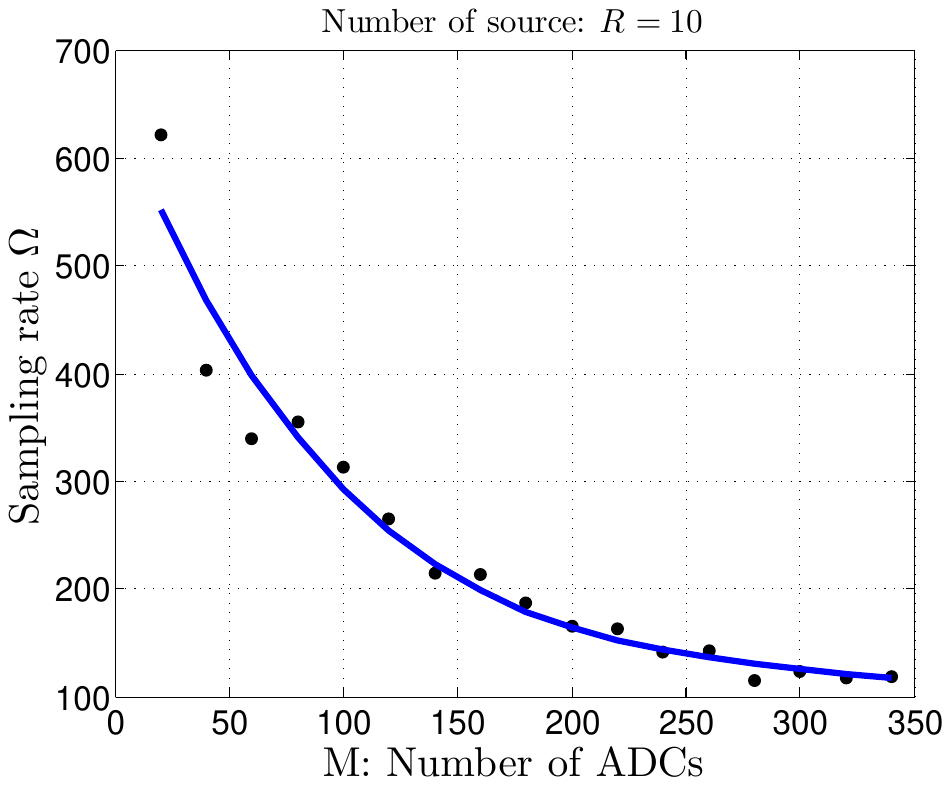}
		\label{fig:RDLvsM}}
	\caption{\small \sl Performance of Sampling Architecture 2. In these experiments, we take an ensemble of $100$ signals, each bandlimited to $512$Hz. The probability of success is computed over 100 iterations. (a) Oversampling factor $\eta$ as a function of the number $R$ of underlying independent signals in $\mX_c(t)$. The blue line is the least-squares fit of the data points. (b) Sampling rate $\O$ versus the number $M$ of recieving antennas. The blue line is the least-sqaures fit of the data points. } 
	\label{fig:phasetransitions}
\end{figure}

\subsection{Stable recovery}
In the second set of experiments, we study the performance of the the recovery algorithm when the measurements are contaminated with additive measurement noise $\vxi$ as in \eqref{eq:noisy-measurements}. We generate noise using the standard Gaussian model: $\vxi \sim \mathcal{N}(0,\sigma^2\mI)$. We select $\delta \leq \sigma(L+\sqrt{L})^{1/2}$; a natural choice as the condition $\|\vxi\|_2 \leq \delta$ holds with high probability.  For the experiments in Figure \ref{fig:StableRec}, we solve the optimization program in \eqref{eq:nuclearnorm_minnoisy}. The plot in Figure \ref{fig:StableRec1} shows the relationship between the signal-to-noise ratio (SNR):
\[
\mbox{SNR(dB)} = 10\log\left(\frac{\|\mX_0\|_{\F}^2}{\|\vxi\|_2^2}\right),
\]
and the realtive error(dB):
\[
\mbox{Relative error (dB)} = 10\log\left(\frac{\|\widehat{\mX}-\mX_0\|_{\F}^2}{\|\mX_0\|_{\F}^2}\right)
\]
for a fixed oversampling factor $\eta = 3.5$. The result shows that the relative error degrades gracefully with decreasing SNR. In the Figure \ref{fig:StableRec2}, the plot depicts relative error as a function of the oversampling factor for a fixed SNR = 40dB. The relative error decrease with increasing sampling rate. 
\begin{figure}[ht]
	\centering
	\subfigure[]{
		\includegraphics[trim=2.5cm 7.5cm 2.5cm 7.5cm,scale = 0.4]{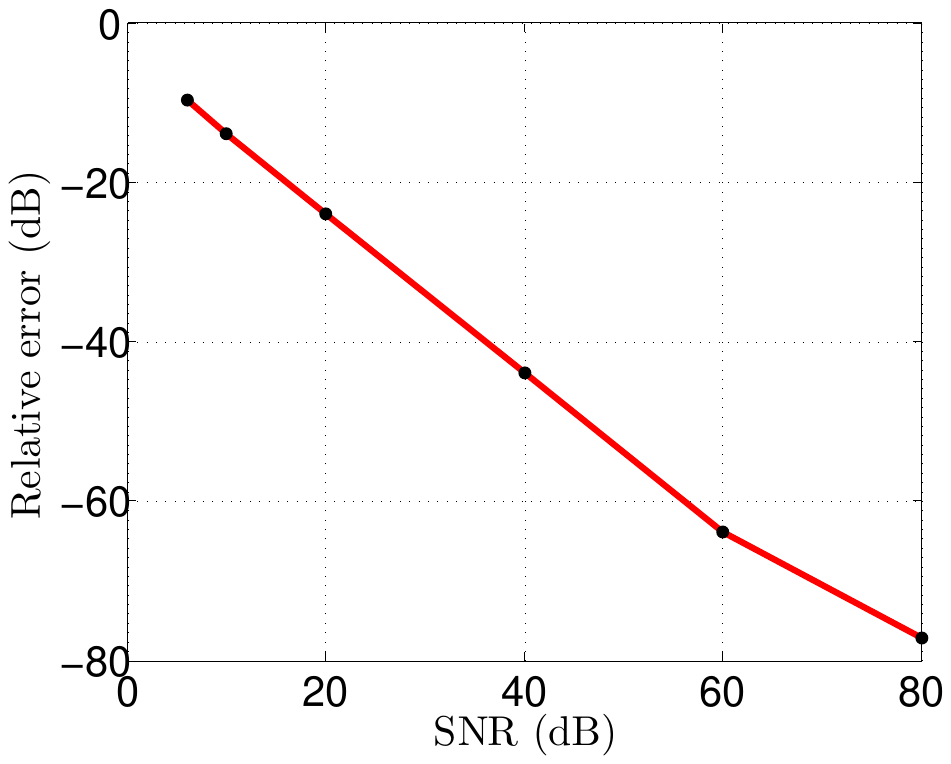}
		\label{fig:StableRec1}}
	\subfigure[]{
		\includegraphics[trim=2.5cm 7.5cm 2.5cm 7.5cm,scale = 0.4]{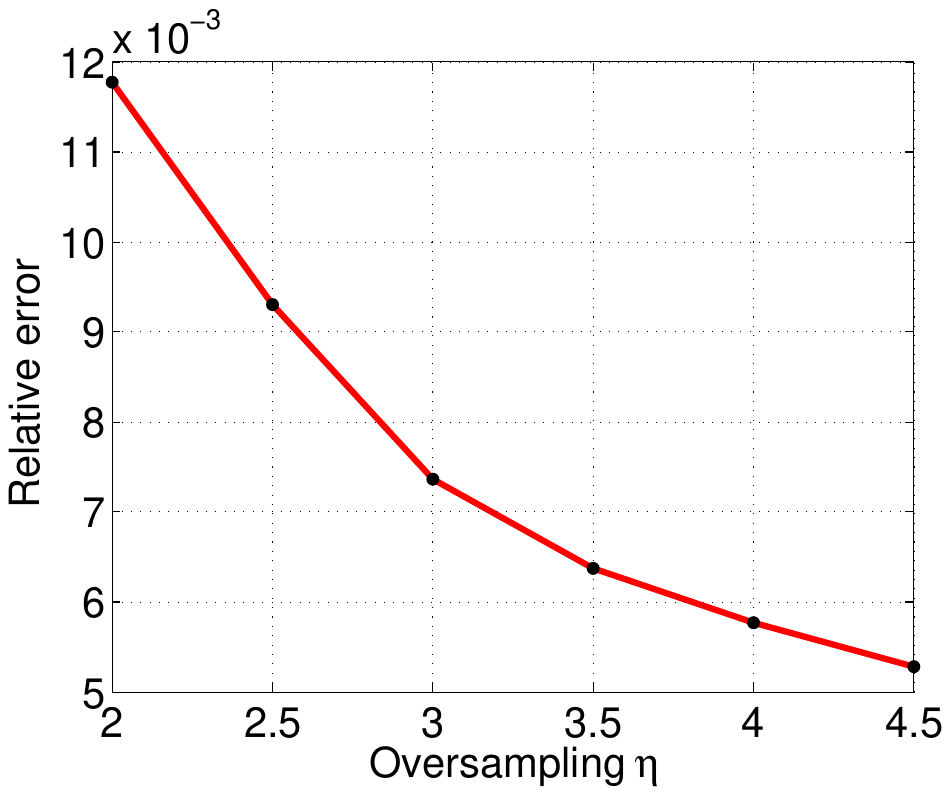}
		\label{fig:StableRec2}}
	\caption{\small \sl Recovery using matix Lasso in the presence of noise.  The input ensemble to the simulated random demodulator consists of $100$ signals, each bandlimited to $512$Hz with number $R = 15$ of latent independent signals.(a) The SNR in dB versus the relative error in dB.  The oversampling factor $\eta = 3.5$. (b) Relative error as a function of the sampling rate. The SNR is fixed at 40dB.}
	\label{fig:StableRec}
\end{figure}

\section{Proof of Lemma \ref{lem:coherence}}\label{sec:coherence-proof}
We start with the proof of Lemma \ref{lem:coherence}
\begin{proof} 
	Recall that $\widetilde{\mU} = \mA\mU$, where we are taking $\mA \in \R^{M \times M}$ to be a random orthogonal matrix, and $\widetilde{\mV} = \mH\mV$, where $\mH$ was defined in \eqref{eq:filter1}. Let $\ve_m$ denote the standard basis vectors in $\R^M$. We begin the proof by noting a standard result; see, \cite{laurent2000adaptive}, that reads
	\begin{equation}\label{eq:Up}
	\underset{m \in [M]}{\max} \|\widetilde{\mU}^*\ve_m\|_2^2 \leq \frac{1}{M}C\beta\max{(R,\log M)},
	\end{equation}
	with probability at least $1-\setO(M^{-\beta})$. Before proving the lemma, we prove an intermediate result:
	\begin{equation}\label{eq:Vp}
	\underset{k \in [W]}{\max} \|\widetilde{\mV}^*\ve_{k}\|_2^2 \leq \frac{1}{W}C\beta\max{(R,\log W)},
	\end{equation}
	where $\ve_{k}$ are standard basis vectors in $\R^W$. 
	Assuming $W$ is even; it will be clear how to extend the argument to $W$ odd.  We can write $\mH = \mW\mQ^*$, where
	\begin{equation}
	\label{eq:filterdecomp1}
	Q[n,\omega] = \begin{cases}
	\frac{1}{\sqrt{W}}&\omega = 0\\
	\frac{2}{\sqrt{W}}\cos \left(\frac{2\pi\omega n}{W}\right)&\omega = [1, \frac{W}{2}-1]\\
	\frac{1}{\sqrt{W}}(-1)^{k-1}&\omega = \frac{W}{2}\\
	\frac{2}{\sqrt{W}}\sin \left(\frac{2\pi\omega n}{W}\right)&\omega = [\frac{W}{2}+1, W-1]\\
	\end{cases} ~~
	W[n,\omega] = \begin{cases}
	\frac{z_0}{\sqrt{W}},&\omega = 0 \\
	\frac{2}{\sqrt{W}}\cos \left(\frac{2\pi\omega n}{W}+\theta_\omega\right)&\omega = [1, \frac{W}{2}-1]\\
	\frac{z_{W/2}}{\sqrt{W}}(-1)^{k-1}, & \omega = \frac{W}{2}\\
	\frac{2}{\sqrt{W}}\sin \left(\frac{2\pi\omega n}{W}+\theta_\omega\right), & \omega = [\frac{W}{2}+1, W-1]
	\end{cases}. 
	\end{equation}
	and $z_0,z_{W/2} = \pm 1$ with equal probability and $\theta_\omega$ for $\omega = 1, \ldots, W/2-1$ are uniform on $[0,2\pi]$ and all $W/2+1$ of these random variables are independent. 
	
	It is a fact that in \eqref{eq:filterdecomp1} for fixed $a$ and $\theta \sim \mbox{Uniform}([0,2\pi])$, the random variables $\mbox{sign}(\cos(a+\theta))$ and $\mbox{sign}(\sin(a+\theta))$ are independent of one another. Thus $\mH$ has the same probability distribution as $\mW\mZ\mQ^*$, where $\mZ = \mbox{diag}(\vz)$ and the entries of $\vz$ are iid $\pm 1$ random variables. In light of this, we will replace $\mH$ with $\mW\mZ\mQ^*$. For a fixed $k$, we can write
	\begin{align*}
	\widetilde{\mV}^*\ve_k &= \mV^*\mH^*\ve_k = \widetilde{\mQ}^*\mZ\vw_k = \sum_{\omega = 1}^W z[\omega]w_k[\omega]\vq_\omega 
	\end{align*} 
	where $\widetilde{\mQ} = \mQ^*\mV$ and $\vw_k = \mW^*\ve_k$ and $\vq_\omega$ is the $\omega$th column of  $\widetilde{\mQ}^*$. We will apply the following concentration inequality.
	\begin{thm}\cite{ledoux2001concentration} Let $\veta \in \R^n$ be a vector whose entries $\eta[i]$ are independent random variables with $\left|\eta[i]\right|< 1$, and let $\mS$ be a fixed $m \times n$ matrix. Then for every $t \geq 0$
		\[
		\PP\big(\left\|\mS\boldsymbol{\eta}\right\|_2 \geq \E\left\|\mS\veta\right\|_2+t\big)\leq 2e^{-t^2/16\left\|\mS\right\|^2}, ~~\text{where} ~~	\E\left\|\mS\veta\right\|_2 \leq \left\|\mS\right\|_{\F}.
		\]
	\end{thm}
	We can apply the above theorem with $\mS = \widetilde{\mQ}^*\mW_k$, where $\mW_k = \mbox{diag}(\vw_k)$, and $\boldsymbol{\eta} = \vz$. In this case, we have 
	\begin{align*}
	\big\|\widetilde{\mQ}^*\mW_k\big\|_{\F}^2 & = \sum_{\omega = 1}^W \left|w_k[\omega]\right|^2\left\|\vq_\omega\right\|_2^2 \leq \frac{2}{W}\sum_{\omega = 1}^W\left\|\vq_\omega\right\|_2^2\leq \frac{2R}{W},
	\end{align*}
	and $\|\widetilde{\mQ}^*\mW_k\| \leq \sqrt{\frac{2}{W}}\|\widetilde{\mQ}\| = \sqrt{\frac{2}{W}}.$ Thus, $
	\PP\bigg(\|\widetilde{\mV}^*\ve_k\|_2> \sqrt{\frac{2R}{W}}+t\sqrt{\frac{2}{W}}\bigg)\leq 2e^{-t^2/16},$ and using the union bound
	\[
	\PP\bigg(\max_{k \in [W]}\|\widetilde{\mV}^*\ve_k\|_2> \sqrt{\frac{2R}{W}}+t\sqrt{\frac{2}{W}}\bigg) \leq 2We^{-t^2/16}.
	\]
	We can make this probability less than $W^{-\beta}$ by taking $t \geq C\sqrt{\log W}$, and \eqref{eq:Vp} follows.
	
	Now to prove \eqref{eq:kappa-bound}, and \eqref{eq:vartheta-bound}, we can write $\mH = \mW\mZ\mQ^*$. Let $\vw_k$ be the $k$th column of $\mW^*$ and let $\tilde{\vu}_m^*$ be the $m$th row of $\widetilde{\mU}$. For a fixed row index $m$ and column index $k$, we can write an entry of $\widetilde{\mU}\widetilde{\mV}^*$ as
	\begin{align*}
	\left[\widetilde{\mU}\widetilde{\mV}^*\right]_{m,k} &= \left[\widetilde{\mU}(\mW\mZ\mQ^*\mV)^*\right]_{m,k} = \left[\widetilde{\mU}\widetilde{\mQ}^*\mZ\mW^*\right]_{m,k}= \tilde{\vu}_m^*\widetilde{\mQ}^*\mZ \vw_k
	\end{align*}
	where $\widetilde{\mQ} = \mQ^*\mV$ is a tall orthonormal matrix. Let $\vp_m^* = \tilde{\vu}_m^*\widetilde{\mQ}^*$.  Since the $z[\omega]$ are iid random variables, a standard applications of the Hoeffding inequality tells us that
	\begin{align*}
	\PP\big(\big|\big[\widetilde{\mU}\widetilde{\mV}^*\big]_{m,k}\big|>\lambda\big)\leq 2e^{-\lambda^2/2\sigma^2}, \quad \mbox{where} \quad 
	\sigma^2 & = \E \vw_k\mZ^*\vp_m\vp_m^* \mZ \vw_k \leq \frac{2\left\|\vp_m\right\|_2^2}{W} = \frac{2\left\|\tilde{\vu}_m\right\|_2^2}{W}.
	\end{align*}
	Thus, with probability exceeding $1-2W^{-\beta}$
	\begin{align}\label{eq:Hoeffding-result}
	\bigg|\left[\widetilde{\mU}\widetilde{\mV}^*\right]_{m,k}\bigg|^2 \leq \frac{4\beta\log W}{W}\left\|\tilde{\vu}_m\right\|_2^2.
	\end{align}
    Taking the maximum over $m \in [M]$, and $k \in [W]$ on both sides, and plugging in the bound in \eqref{eq:Up} shows that
    \begin{align}\label{eq:UpVp}
    \frac{MW}{R}\max_{m \in [M]}\max_{k \in [W]}~\left|\ve_m^*\widetilde{\mU}\widetilde{\mV}^*\ve_k\right|^2 \leq C_\beta\log W\max(1,\frac{1}{R}\log M)
    \end{align}
    holds with probability at least $1-\mathcal{O}(W^{-\beta}+M^{-\beta}) = 1-\setO(W^{-\beta})$, where the equality follows from the fact that $W \geq M$. This proves the first claim \eqref{eq:kappa-bound} in Lemma \ref{lem:coherence}. 
    
    Similarly, \eqref{eq:Hoeffding-result} implies that 
    \begin{align}
    \sum_{k \in \setB_\ell} \big|\big[\widetilde{\mU}\widetilde{\mV}^*\big]_{m,k}\big|^2 \leq |\setB_\ell| \frac{4\beta\log W}{W}\left\|\tilde{\vu}_m\right\|_2^2 = \frac{4\beta\log W}{\Omega}\|\tilde{\vu}_m\|_2^2,
    \end{align}
     where $\setB_\ell$ is defined in \eqref{eq:setBl}, and the last equality follows from the fact that $|\setB_\ell| = W/\Omega$. Finally, evaluating the maximum over $m \in [M]$, and $\ell \in [\Omega]$ on both sides, and using the bound in \eqref{eq:Up} shows that 
     \[
     \frac{M\Omega}{R}\max_{m \in [M]} \max_{\ell \in [\Omega]}\|\ve_m^*\widetilde{\mU}\widetilde{\mV}^*\mI_{\setB_\ell}\|_2^2 \leq C_\beta \log W \max(1,\frac{1}{R}\log M),
     \]
    which proves the second claim \eqref{eq:vartheta-bound} in Lemma \ref{lem:coherence}. 
\end{proof}

\section{Proof of Theorem \ref{thm:Exact-rec}}\label{sec:Proof-Exact-rec}

\subsection {Preliminaries}

Recall from \eqref{eq:meas}, we obtain measurements $y_n[\ell]$ of an unknown low-rank matrix $\mX_0$ through a random $M \times W$ rank-1 measurement ensemble $\va_n\vd_{n\ell}^*,~(n,\ell) \in [N]\times[\O]$, where $\va_n^* \in \R^M$ denote the rows of mixing matrix $\mA\in \R^{N \times M}$, and $\vd_{n\ell} \in \R^W$ are random binary on support set $\setB_\ell$, and zero elsewhere. In addition, the vectors$\vdash_{n\ell}$ are independently generated for every $n$, and $\ell$. In Theorem \ref{thm:Exact-rec}, the AVMM simply replicates (without mixing) $P$ copies of $M$ input signals to produce $N$ output signals. This amounts to choosing 
\begin{align}\label{eq:A-thm1}
\mA = \tfrac{1}{\sqrt{P}}[\mI_M\otimes \mathbf{1}_P] = \tfrac{1}{\sqrt{P}}[\mI_M,\ldots, \mI_M] ^*,
\end{align}
 where $P = N/M$. From this construction, we have $\|\va_n\|_2^2 = 1/P$, and $\mA^*\mA = \mI_M$. Also recall that using the definition of linear map $\setA$ in \eqref{eq:cA-def}, the measurements are compactly expressed as $\vy = \setA(\mX_0).$
Moreover, the adjoint operator is  
\[
\setA^*(\vy) = \sum_{n\in [N]}\sum_{\ell \in [\Omega]} y_{n}[\ell]\va_n\vd_{n\ell}^* = \sum_{n\in [N]}\sum_{\ell \in [\Omega]} \va_n\va_n^*\mX_0\vd_{n\ell}\vd_{n\ell}^*,
\]
where the second equality is the result of \eqref{eq:meas}. It will also be useful to visualize the linear operator $\setA^*\setA$ in a matrix form:
\begin{equation}\label{eq:cAtcA}
\setA^*\setA = \sum_{n \in [N]}\sum_{\ell \in [\Omega]} \va_n\vd_{n\ell}^* \boxtimes \va_n\vd_{n\ell}^* = \sum_{n\in[N]}\sum_{\ell \in [\Omega]} \va_n\va_n^* \otimes \vd_{n\ell}\vd_{n\ell}^*,
\end{equation}
where $\boxtimes$ denotes the tensor product.  In general, the tensor product of rank-1 matrices $\vx_1\vy_1^*$, $\vx_2\vy_2^*$ with $\vx_i \in \R^M$, and $\vy_i \in \R^W$ is given by the big matrix
\[
\vx_1\vy_1^*\boxtimes\vx_2\vy_2^* = 
\begin{bmatrix}
x_1[1]x_2[1]\vy_1\vy_2^* & x_1[1]x_2[2]\vy_1\vy_2^* & \cdots & x_1[1]x_2[M]\vy_1\vy_2^* \\
x_1[2]x_2[1]\vy_1\vy_2^* & x_1[2]x_2[2]\vy_1\vy_2^* & \cdots & x_1[2]x_2[M]\vy_1\vy_2^*\\
\vdots & & \ddots & \\
x_1[M]x_2[1]\vy_1\vy_2^* & \cdots & \cdots & x_1[M]x_2[M]\vy_1\vy_2^*
\end{bmatrix}.
\]
With this definition it is easy to visualize that $\E\setA^*\setA = \mathcal{I}$. Let $\{\vu_k^*\}_{  k \in [M]}$, and $\{\vv_k^*\}_{k \in [W]}$ denote the rows of the matrices $\mU$, and $\mV$, respectively. 

We begin by defining a subspace $T \subset \R^{M \times W}$ associated with $\mX_0$ with singular-value decomposition given by  $\mX_0 = \mU\mSigma\mV^*$ 
\[
T = \{\mX: \mX = \mU\mZ_1^* + \mZ_2\mV^* , \mZ_1 \in \R^{W\times R}, \mZ_2 \in \R^{M \times R}\}.
\]
The orthogonal projections onto $T$, and its orthogonal complement $T^\perp$ are defined as $\PT(\mZ) = \mU\mU^* \mZ + \mZ\mV\mV^*  - \mU\mU^*\mZ\mV\mV^*$, and $\PTc(\mZ) = \mZ- \PT(\mZ)$, respectively. In the proofs later, we repeatedly make use of the following calculation
\begin{align*}
\left\|\PT(\va_n\vd^*_{n\ell})\right\|_{\F}^2  &= \< \PT (\va_n\vd^*_{n\ell}) , \va_n\vd^*_{n\ell}\>\\
& = \< \mU^*\va_n\vd^*_{n\ell} , \mU^*\va_n\vd^*_{n\ell}\> + \< \va_n\vd^*_{n\ell}\mV , \va_n\vd^*_{n\ell}\mV\> - \< \mU^*\va_n\vd^*_{n\ell}\mV, \mU^*\va_n\vd^*_{n\ell}\mV \>\notag\\
&= \|\mU^*\va_n\vd^*_{n\ell}\|_{\F}^2 + \|\va_n\vd^*_{n\ell}\mV\|_{\F}^2 - \|\mU^*\va_n\vd^*_{n\ell}\mV\|_{\F}^2\leq \|\mU^*\va_n\vd^*_{n\ell}\|_{\F}^2 + \|\va_n\vd^*_{n\ell}\mV\|_{\F}^2. 
\end{align*}
Observe that 
\[
\|\mU^*\va_n\vd^*_{n\ell}\|_{\F}^2  = \|\mU^*\va_n\|_2^2\|\vd^*_{n\ell}\|_2^2 = \frac{W}{\Omega}\|\mU^*\va_n\|_2^2, ~ \text{and} ~ \|\va_n\vd_{n\ell}^*\mV\|_{\F}^2 = \|\va_n\|_2^2 \|\vd_{n\ell}^*\mV\|_2^2 =  \frac{1}{P}\|\vd_{n\ell}^*\mV\|_2^2.
\]
This leads us to 
\begin{align}
\|\PT(\va_n\vd^*_{n\ell})\|_{\F}^2 &\leq \frac{W}{\O} \|\mU^*\va_n\|_2^2 + \frac{1}{P}\|\vd_{n\ell}^*\mV\|_2^2. \label{eq:PTA-Fro-norm}
\end{align}
Finally, we will also require a bound on the operator norm of the linear map $\setA$. To this end, note that the measurement matrices $\va_n\vd_{n\ell}^*$ are orthogonal for every $\ell \in [\O]$ in the standard Hilbert-Schmidt inner product, that is, $\<\va_n\vd^*_{n\ell},\va_n\vd^*_{n\ell^\prime}\>=0$  whenever $\ell \neq \ell^\prime$. This directly implies a  following bound on the operator of $\setA$:
\begin{equation}\label{eq:cA-operatornorm}
\|\setA\| \leq \sqrt{\sum_{n \in [N]} \|\va_n\vd_{n\ell}^*\|_{\F}^2} = \sqrt{\frac{MW}{\O}} \leq W,
\end{equation}
where in the last inequality we used the fact that $M \leq W$, and $\Omega \geq 1$. Although a much tighter bound can be achieved using results from random matrix theory, the loose bound is sufficient for our purposes. 
\subsection{Sufficient condition for the uniqueness}
Uniqueness of the minimizer to \eqref{eq:nuclearnorm_min} can be guaranteed by the sufficient condition \cite{candes09ex,gross11re} given below. 
\begin{prop}\label{prop:suff-cond}
	The matrix $\mX$ is the unique minimizer to \eqref{eq:nuclearnorm_min} if $\exists \mY \in \mbox{Range}(\setA^*)$ such that $\forall \mZ\in \mbox{Null}(\setA)$ 
	\begin{equation*}
	\left(1-\|\PTc(\mY)\|\right)\|\PTc(\mZ)\|_{*}-\|\mU\mV^*-\PT(\mY)\|_{\F}\|\PT(\mZ)\|_{\F} > 0.
	\end{equation*}
\end{prop}
In light of the proposition, it is sufficient to show that  $\exists \mY \in \mbox{Range}(\setA^*)$, such that  
\begin{equation}\label{eq:sufficient-condition-1}
\|\PT(\mY)-\mU\mV^*\|_{\F} \leq 1/3W, \quad \|\PTc(\mY)\| \leq 1/2, 
\end{equation}
and  for every $\mZ \in \mbox{Null}(\setA)$, 
\begin{align}\label{eq:sufficient-condition-2}
\|\PTc(\mZ)\|_{\F} \geq (1/\sqrt{2}W)\|\PT(\mZ)\|_{\F}
\end{align}
 holds. This can be immediately shown as follows $
0  =  \|\setA(\mZ)\|_{\F} \geq  \|\setA(\PT(\mZ))\|_{\F} - \|\setA(\PTc(\mZ))\|_{\F} \geq \|\setA(\PT(\mZ))\|_{\F} - W\|\PT^\perp(\mZ)\|_{\F}.$ In addition, for an arbitrary $\mZ$, we have 
\begin{align*} 
\|\setA(\PT(\mZ))\|_{\F}^2 &= \left\langle \setA(\PT(\mZ)), \setA(\PT(\mZ))\right\rangle =\left\langle \mZ, \PT\setA^*\setA\PT(\mZ)\right\rangle\notag\\
&\geq (1- \|\PT\setA^*\setA\PT-\PT\|)\|\PT(\mZ)\|_{\F}^2\geq \frac{1}{2} \|\PT(\mZ)\|_{\F}^2,
\end{align*}
where the last inequality is obtained by plugging in $\|\PT\setA^*\setA\PT-\PT\| \leq 1/2$, which will be shown to be true under appropriate choice of $\Omega$ with probability at least $1-\mathcal{O}(W^{-\beta})$ in Corollary \ref{cor:injectivity}. Combining the last two inequalities gives us the result in \eqref{eq:sufficient-condition-2}.

\subsection{Golfing scheme for the random modulator}
For technical reasons, we will work with partial linear maps $\setA_p: \R^{M \times W} \rightarrow \R^{M\O},~p\in[P]$ modified from the linear map $\setA$ in \eqref{eq:cA-def}.  Define $P$ partitions $\{\Gamma_p\}_{p=1}^P$ of the index set $[N]$ as $\Gamma_p := \{(p-1)M+1,\ldots, pM\}$ for every $p \in [P]$. Clearly, $\Gamma_p \cap \Gamma_{p^\prime} = \emptyset$, and $\cup_{p=1}^P \Gamma_p = [N]$. We will take the number of partitions\footnote{We assume that $N/M$ is an integer--this can be ensured in the worst case by doubling $N$.} $P = N/M$.
The partial linear maps $\setA_p$ are defined as 
\begin{equation}\label{eq:cAp}
\setA_p(\mX) := \{\va^*_n\mX\vd_{n\ell},~ n \in \Gamma_p, ~\ell \in [\O]\}.
\end{equation}
Using the definition of $\mA$ in \eqref{eq:A-thm1}, it is clear that 
\begin{align}\label{eq:a_n-e_m}
\{\va_n: n \in \Gamma_p\} = \{\tfrac{1}{\sqrt{P}}\ve_m: m \in [M] \} ~ \mbox{for every} ~p \in [P] \implies \sum_n \va_n\va_n^* = \tfrac{1}{P}\mI_M.
\end{align}
The corresponding adjoint operator maps a vector $\vz \in \R^{M\O}$ to an $M \times W$ matrix 
\[
\setA_p^*(\vz) = \underset{n \in \Gamma_p}{\sum}\sum_{\ell \in [\O]} z_{n}[\ell]\va_n\vd_{n\ell}^*.
\]
It will also be useful to make a note of the following versions of the above definition
 \begin{align}\label{eq:cAp-T-cAp}
 \setA_p^*\setA_p(\mX) &= \underset{n \in \Gamma_p}{\sum}\sum_{\ell \in [\O]} \va_n\va_n^*\mX\vd_{n\ell}\vd_{n\ell}^*, ~\text{and} ~ \setA_p^*\setA_p = \underset{n \in \Gamma_p}{\sum}\sum_{\ell \in [\O]}\va_n\vd_{n\ell}^*\boxtimes\va_n\vd_{n\ell}^*, 
 \end{align}
where the second definition just emphasizes the fact that the linear map $\setA_p^*\setA_p$ can be thought of as a big $MW \times MW$ matrix that operates on a vectorized $\mX$. 

With the linear operators defined on the subsets $\{\Gamma_p\}_{p=1}^P$ above, we write the iterative construction of the dual certificate:
\begin{align}\label{eq:iterative-construction-Y}
\mY_p = \mY_{p-1}-\setA_p^*\setA_p\big(\PT(\mY_{p-1})-\mU\mV^*\big), \quad \text{where} \quad \mY_p \in \mbox{Range}(\setA^*),
\end{align}
where we take $\mY_0 = \mathbf{0}$.  Projecting onto the subspace $T$ on both sides results in $
\PT(\mY_p) = \PT(\mY_{p-1})-\PT\setA_p^*\setA_p(\PT(\mY_{p-1})-\mU\mV^*)$.
Define
\begin{align}\label{eq:Wk-def}
\mW_p :& = \PT (\mY_p) -\mU\mV^*,
\end{align}
the iteration takes the equivalent form $\mW_p = \mW_{p-1}-\PT\setA_p^*\setA_p\PT(\mW_{p-1})$. We will take $\mY = \mY_P$ to be our candidate for the dual certificate and the rest of this section concerns showing that $\mY_P$ obeys the conditions in \eqref{eq:sufficient-condition-1}. Let's start by showing that $\|\PT(\mY_P)-\mU\mV^*\|_{\F} \leq (3W)^{-1}$ holds. To this end, note that from the iterative construction above, the following bound immediately follows
\begin{align*}
\left\|\mW_p\right\|_{\F} &\leq \|\PT\setA_p^*\setA_p\PT-\PT\|\|\mW_{p-1}\|_{\F}.
\end{align*}
From Lemma \ref{lem:injectivity}, we have $\|\PT\setA_p^*\setA_p\PT-\PT\| \leq 1/2$ for every $p \in [P]$. This means that $\|\mW_p\|_{\F}$ cuts after every iteration giving us the following bound on the Frobenius norm of the final iterate $\mW_P$
\begin{align}\label{eq:P-bound}
\left\|\mW_P\right\|_{\F} &\leq 2^{-P}\|\mU\mV^*\|_{\F} = 2^{-P}\sqrt{R} \leq (3W)^{-1} \quad\mbox{when}\quad P \geq  2\log_2 (3W).
\end{align}
Using the union bound over $p \in [P]$, the bound on $\|\mW_P\|_{\F}$ holds with probability at least $1-\mathcal{O}(PW^{-\beta})\geq 1-\mathcal{O}(W^{1-\beta}) $. This proves that the candidate dual certificate $\mY_P$ obeys the first condition in \eqref{eq:sufficient-condition-1}.

Since $P = N/M$, this implies that the number $N$ of output channels from the analog-vector-matrix multiplier in Figure \ref{fig:Rand-dem} must be a factor of roughly $\log W$ compared to the input channels, i.e., 
\begin{equation}\label{eq:Nbound}
N \geq CM\log W.
\end{equation}
However, we believe this requirement is merely an artifact of using golfing scheme as the proof strategy for Theorem \ref{thm:Exact-rec}. In practice, all our simulations point to $N = M$, that is, the number of channels at the output of the AVMM are equal to the input channels.

From the iterative construction \eqref{eq:iterative-construction-Y}, it is clear that $\mY_P = -\sum_{p = 1}^P \setA_p^*\setA_p (\mW_{p-1})$. We will now converge on showing that $\mY_P$ satisfies the second condition in \eqref{eq:sufficient-condition-1}. Begin with 
\begin{align*}
\left\|\PTc(\mY_P)\right\| &\leq \sum_{p = 1}^P \left\|\PTc(\setA_p^*\setA_p(\mW_{p-1}))\right\| = \sum_{p = 1}^P \left\|\PTc\left(\setA_p^*\setA_p(\mW_{p-1})-\mW_{p-1}\right)\right\|, 
\end{align*}
where the last equality follows from the fact that $\mW_{p-1} \in T$. Since $\|\PTc\|\leq 1$, we have 
\begin{align*}
\sum_{p = 1}^P \left\|\PTc\left(\setA_p^*\setA_p(\mW_{p-1})-\mW_{p-1}\right)\right\| \leq \sum_{p = 1}^P \left\|\setA_p^*\setA_p(\mW_{p-1})-\mW_{p-1}\right\| \leq \sum_{p=1}^P 2^{-p-1}  < 1/2,
\end{align*}
the second last inequality above requires $\|\setA_p^*\setA_p(\mW_{p-1})-\mW_{p-1}\| \leq 2^{-p-1}, ~~ \text{for every}~p \in [P]$, which using Lemma \ref{lem:JL} is only true when $\Omega \geq C_\beta\nu^2(R/M)W\log^2W$ with probability at least $1-\mathcal{O}(PW^{-\beta})\geq 1-\mathcal{O}(W^{-\beta+1})$, where the factor $P$ comes from the union bound over every $p \in [P]$. Lemma \ref{lem:JL} 

Combining sample complexities in \eqref{eq:sample-complexity-injectivity}, and \eqref{eq:sample-complexity-JL}, and using the definition of $\vartheta^2$ in \eqref{eq:max-coherence-def} gives us the proof of Theorem \ref{thm:Exact-rec}.


\subsection{Key Lemmas}
We now state the key lemmas to prove Theorem \ref{thm:Exact-rec}.
\begin{lem}\label{lem:injectivity}
	Fix $\beta \geq 1$. Assume that 
	\begin{align}\label{eq:sample-complexity-injectivity}
	\O \geq C_\beta \max(\mu^2,\rho^2) \frac{R}{M}W\log^2W,
	\end{align}
	 where $C_\beta$ is a universal constant only depending on $\beta$. Then the linear operator $\setA_p$ obeys
	 \[\left\|\PT\setA_p^*\setA_p\PT-\PT\right\|\leq \frac{1}{2}
	 \]
	  with probability at least  $1- \mathcal{O}(W^{-\beta})$.
\end{lem}
Proof of this lemma will be presented in Section \ref{sec:injectivity}.
\begin{cor}\label{cor:injectivity}
	Fix $\beta \geq 1$. Assume $\O \geq C_\beta \max(\mu^2,\rho^2) \tfrac{R}{M}W\log^2W$, where $C_\beta$ is a universal constant that only depends on $\beta$. Then the linear operator $\setA$ defined in \eqref{eq:cA-def} obeys  $\|\PT\setA^*\setA\PT-\PT\|\leq 2^{-1}$ with probability at least $1- \mathcal{O}(W^{-\beta})$.
\end{cor}
\begin{proof}
	Proof of this corollary follows exactly the same steps as the proof of Lemma \ref{lem:injectivity} with only difference being that we take $P = N/M = 1$. 
\end{proof}
\begin{lem}\label{lem:coherence-iterates}
	Define coherence $\nu_{p}^2$ of the iterates $\mW_p$ as 
	\begin{equation}\label{eq:nup}
	\nu_p^2 : = \frac{M\O}{R}\max_{m \in [M]}\max_{ \ell \in [\Omega]}\|\ve_m^*\mW_{p}\mI_{\setB_\ell}\|_{\F}^2.
	\end{equation}
	Then under the same conditions as in \eqref{eq:sample-complexity-injectivity}, we have $\nu_{p} \leq (1/2)\nu_{p-1}$ with probability at least $1-\mathcal{O}(W^{-\beta})$. 
\end{lem}
\begin{proof}
	The proof of this lemma follows similar techniques and matrix Bernstein inequality as used in Lemma \ref{lem:injectivity}. Similar results can be found in \cite{ahmed2012blind}. We skip the proof due to space constraints. 
\end{proof}
Using the definition of $\nu^2$ in \eqref{eq:coherence}, and the fact that $\mW_0 = -\mU\mV^*$, we can see that $\nu_0^2 =  \nu^2$.  Invoking Lemma \ref{lem:coherence-iterates} for every $p \in [P]$, we can iteratively conclude that 
\begin{equation}\label{eq:coherence-bound}
\nu_P \leq 2^{-P}\nu
\end{equation}
with probability at least $1-\mathcal{O}(PW^{-\beta}) \geq 1-\mathcal{O}(W^{1-\beta})$. 
\begin{lem}\label{lem:JL}
	Fix $\beta \geq 1$. Take 
	\begin{align}\label{eq:sample-complexity-JL}
	\O \geq C\beta \nu^2 \frac{R}{M}W\log^2 W
	\end{align}
	 for a sufficiently large constant $C$. Let $\mW_{p-1}$ be a fixed $M \times W$ matrix defined in \eqref{eq:Wk-def} then
	\[
	\left\|\setA_p^*\setA_p(\mW_{p-1})- \mW_{p-1}\right\| \leq 2^{-p-1}
	\]
	with probability at least $1-\mathcal{O}(W^{-\beta})$.
\end{lem}
Proof of this lemma will be presented in Section \ref{sec:JL}
\section{Proof of Key Lemmas}
Proof of all the key lemmas mainly relies on using matrix Bernstein inequality to control the operator norms of sums of random matrices. 
\subsection{Matrix Bernstein-type inequality}
We will use a specialized version of the matrix Bernstein-type inequality \cite{tropp12us,koltchinskii10nu} that depends on the Orlicz norms. The Orlicz norm of a random matrix $\mZ$ is defined as 
\begin{equation}\label{eq:matpsinorm}
\|\mZ\|_{\psi_\alpha} = \inf \{ u>0: \E \exp (\frac{\|\mZ\|^\alpha}{u^\alpha}) \leq 2\}, \quad \alpha \geq 1.
\end{equation}
Suppose that, for some constant $U_{\alpha} > 0, \|\mZ_q\|_{\psi_\alpha} \leq U_{(\alpha)}, q = 1, \ldots, Q$ then the following proposition holds. 
\begin{prop}\label{prop:matbernpsi}
	Let $\mZ_1,\mZ_2,\ldots,\mZ_Q$ be iid random matrices with dimensions $M \times N$ that satisfy $\E (\mZ_q) = 0$. Suppose that $\|\mZ\|_{\psi_\alpha} < \infty$ for some $\alpha \geq 1 $. Define 
	\begin{align}\label{eq:matbernsigma}
	&\sigma_Z  = \max \left\{\left\|\sum_{q = 1}^Q (\E \mZ_q\mZ_q^*)\right\|^{1/2},\left\|\sum_{q = 1}^Q (\E \mZ_q^*\mZ_q)\right\|^{1/2} \right\}
	\end{align}
	Then $\exists$ a constant $C > 0$ such that , for all $t>0$, with probability at least $1-\mathrm{e}^{-t}$
	\begin{align}\label{eq:matbernpsi}
	&\left\|\mZ_1+\cdots+\mZ_Q\right\| \leq C \max\left\{\sigma_Z\sqrt{t+\log(M+N)},U_{\alpha}\log^{1/\alpha}\left(\frac{QU_\alpha^2}{\sigma_Z^2}\right)(t+\log(M+N))\right\}
	\end{align}
\end{prop}
\subsection{Proof of Lemma \ref{lem:injectivity}}\label{sec:injectivity}
We start by writing $\PT\setA_p^*\setA_p\PT$ as a sum of independent random matrices using \eqref{eq:cAp-T-cAp} to obtain 
\[
\PT\setA_p^*\setA_p\PT  = \underset{n \in \Gamma_p}{\sum}\sum_{\ell \in [\O]} P\left[\PT (\va_n\vd_{n\ell}^*)\boxtimes \PT(\va_n\vd_{n\ell}^*)\right].
\]
Using \eqref{eq:a_n-e_m},  and the fact that $\sum_{\ell \in [\Omega]}\E \vd_{n\ell}\vd_{n\ell}^* = \mI_W$, the expectation of the quantity above evaluates to
\[
\E \PT\setA_p^*\setA_p \PT = \PT\bigg[P \underset{n \in \Gamma_p}{\sum}\sum_{\ell \in [\O]}  \E (\va_n\va_n^* \otimes \vd_{n\ell}\vd_{n\ell}^*) \bigg]\PT = \PT \bigg[ P\sum_{n \in \Gamma_p}\va_n\va_n^* \otimes \sum_{\ell \in [\Omega]}\E \vd_{n\ell}\vd_{n\ell}^*\bigg]\PT= \PT.
\]
The quantity $\PT\setA_p^*\setA_p\PT- \PT$ can therefore be expressed as a sum of independent zero mean random matrices in the following form 
\begin{align*}
\PT\setA_p^*\setA_p\PT-\PT = \underset{n \in \Gamma_p}{\sum}\sum_{\ell \in [\O]}  P\left[\PT (\va_n\vd_{n\ell}^*)\boxtimes \PT(\va_n\vd_{n\ell}^*)- \E \PT (\va_n\vd_{n\ell}^*)\boxtimes \PT(\va_n\vd_{n\ell}^*)\right]. 
\end{align*}
We will employ matrix Bernstein inequality to control the operator norm of the above sum. To proceed define the operator $\mathcal{Z}_{n\ell}$ which maps $\mZ$ to $\left\langle \PT(\va_n\vd_{n\ell}^*),\mZ\right\rangle \PT(\va_n\vd_{n\ell}^*)$, i.e., $\mathcal{Z}_{n\ell} = \left[\PT (\va_n\vd_{n\ell}^*)\boxtimes \PT(\va_n\vd_{n\ell}^*)\right]$. This operator is rank one, therefore, the operator norm $\left\|\mathcal{Z}_{n\ell}\right\| = \left\|\PT(\va_n\vd_{n\ell}^*)\right\|_{\F}^2$. To ease the notation, we will use $\underset{n,\ell}{\sum}$ as a shorthand for $\underset{n \in \Gamma_p}{\sum}\underset{\ell \in [\O]}{\sum} $. 
We begin by computing the variance in \eqref{eq:matbernsigma} as follows
\begin{align*}
\sigma^2 &= P^2\bigg\|\underset{n,\ell}{\sum}\E (\mathcal{Z}_{n\ell} -\E\mathcal{Z}_{n\ell} )^2 \bigg\|= P^2\bigg\|\underset{n,\ell}{\sum}[\E \mathcal{Z}^2_{n\ell} -(\E\mathcal{Z}_{n\ell} )^2]\bigg\| \leq P^2\bigg\|\underset{n,\ell}{\sum}\E \mathcal{Z}^2_{n,\ell}\bigg\|,
\end{align*}
where the last inequality follows from the fact that $\E (\mathcal{Z}^2_{n\ell})$, and $(\E\mathcal{Z}_{n\ell} )^2$ are symmetric, and positive-semidefinite matrices.  The square of the rank-1 matrices $\mathcal{Z}_{n\ell}$ is simply given by  $\mathcal{Z}^2_{n\ell} = \|\PT(\va_n\vd_{n\ell}^*)\|_\F^2 \mathcal{Z}_{n\ell}$. 
Now we develop the operator norm of the result simplified expression using \eqref{eq:PTA-Fro-norm}
\begin{align*}
\bigg\|\E \underset{n,\ell}{\sum}\left\|\PT(\va_n\vd_{n\ell}^*)\right\|_{\F}^2 \mathcal{Z}_{n\ell}\bigg\| &\leq \bigg\|\underset{n,\ell}{\sum}\bigg(\frac{W}{\O}\|\mU^*\va_n\|_2^2 +  \frac{1}{P}\|\mV^*\vd_{n\ell}\|_2^2 \bigg) \mathcal{Z}_{n\ell}\bigg\|.
\end{align*}
Using the definition in \eqref{eq:coherence}, and \eqref{eq:a_n-e_m}, we can bound $\|\mU^*\va_n\|_2^2 \leq \mu^2R/PM$. Using this fact, we have
\begin{align}
\sigma^2 &\leq \bigg[P\mu^2\frac{RW}{M\Omega} \bigg\|\underset{n,\ell}{\sum}\E\mathcal{Z}_{n\ell} \bigg\|+P\bigg\| \underset{n,\ell}{\sum}\E\|\mV^*\vd_{n\ell}\|_2^2\mathcal{Z}_{n\ell}\bigg\|\bigg] \leq \bigg[\mu^2\frac{RW}{M\O} + P\bigg\|\underset{n,\ell}{\sum}\E\|\mV^*\vd_{n\ell}\|_2^2\mathcal{Z}_{n\ell}\bigg\|\bigg].\label{eq:var-main}
\end{align}
The second term in \eqref{eq:var-main} can be simplified as 
\begin{align}
\bigg\|\underset{n,\ell}{\sum}\E\|\mV^*\vd_{n\ell}\|_2^2\mathcal{Z}_{n\ell}\bigg\|&= \bigg\|\PT \underset{n,\ell}{\sum}\left(\E\|\mV^*\vd_{n\ell}\|_2^2 (\va_n\vd_{n\ell}^* \boxtimes \va_n\vd_{n\ell}^*)\right)\PT\bigg\|\notag\\
&\leq \bigg\| \underset{n,\ell}{\sum}\left(\E\|\mV^*\vd_{n\ell}\|_2^2 (\va_n\vd_{n\ell}^* \boxtimes \va_n\vd_{n\ell}^*)\right)\bigg\|\label{eq:var-sub1},
\end{align}
where the last inequality follows form the fact that $\|\PT\| \leq 1$. Since $\va_n\vd_{n\ell}^* \boxtimes \va_n\vd_{n\ell}^* = \va_n\va_n^* \otimes \vd_{n,\ell}\vd_{n,\ell}^*$, and a simple calculation reveals the expectation 
\[
\E\|\mV^*\vd_{n\ell}\|_2^2 ( \vd_{n\ell}\vd_{n\ell}^*) = \big[  \|\mV^*\mI_{\setB_\ell}\|_{\F}^2\mI_{\setB_\ell}+2\mI_{\setB_\ell}\mV\mV^*\mI_{\setB_\ell}-2\text{diag}(\mI_{\setB_\ell}\mV\mV^*\mI_{\setB_\ell})\big] \preccurlyeq 3 \|\mV^*\mI_{\setB_\ell}\|_{\F}^2 
\]
where for $\text{diag}(\mX)$ is the diagonal matrix obtained by setting the off-diagonal entries of $\mX$ to zero, and $\mI_{\setB_\ell}$ denotes the $W \times W$ identity matrix with ones only at the diagonal positions indexed by $\setB_\ell$. This directly implies that 
\begin{align}\label{eq:Variance-sub}
\bigg\|\sum_{n,\ell}\E\|\mV^*\vd_{n\ell}\|_2^2 (\va_n\vd_{n\ell}^* \boxtimes \va_n\vd_{n\ell}^*)\bigg\|& \leq \bigg\|\bigg[\sum_{\ell \in [\Omega]}  3\|\mV^*\mI_{\setB_\ell}\|_{\F}^2\mI_{\setB_\ell}\bigg] \otimes \bigg[\sum_{n \in \Gamma_p} \va_n\va_n^*\bigg]   \bigg\|\notag\\
& \leq \frac{3}{P}\max_{\ell \in [\O]} \|\mV^*\mI_{\setB_\ell}\|_{\F}^2 = 3\rho^2 \frac{R}{P\O}, 
\end{align}
where the last equality follows from the definition of the coherence in \eqref{eq:coherence}.
Plugging \eqref{eq:Variance-sub} in \eqref{eq:var-main}, we have the bound 
\begin{align}
\sigma^2 &\leq C\left(\mu^2\frac{ RW}{M\O} +\rho^2\frac{R}{\O}\right).\label{eq:Lem1-var}
\end{align}
Finally, we calculate the Orlicz norm, the last ingredient to obtain the Bernstein bound. First, it is important to see that 
\begin{equation*}
P\|\mathcal{Z}_{n\ell} - \E \mathcal{Z}_{n\ell}\| \leq 2P\|\mathcal{Z}_{n\ell}\| = 2P\|\mathcal{Z}_{n\ell}\|_{\F} = 2P\|\PT(\va_n\vd_{n\ell}^*)\|_{\F}^2,
\end{equation*}
where the second-last equality follows form the fact that $\mathcal{Z}_{n\ell}$ is the rank-1 operator. Using the last equation, and \eqref{eq:PTA-Fro-norm}, we have 
\begin{align}
U_1&= \max_{n\in[N]} \max_{\ell \in [\Omega]}2P\big\|\|\PT(\va_n\vd_{n\ell}^*)\|_{\F}^2\big\|_{\psi_1}  \leq 2P\max_{n\in[N]} \max_{\ell \in [\Omega]} \bigg\|\frac{W}{\O}\|\mU^*\va_n^*\|_2^2 + \frac{1}{P}\|\mV^*\vd_{n\ell}\|_2^2\bigg\|_{\psi_1}\notag\\
&\leq CP\left(\mu^2 \frac{R}{PM}\frac{W}{\O} + \frac{1}{P}\max_{\ell \in [\Omega]}\|\mV^*\mI_{\setB_\ell}\|_{\F}^2\right) \leq  C\left(\mu^2\frac{RW}{M\O}+\rho^2\frac{R}{\O}\right)\label{eq:orlicz-bound}
\end{align}
Moreover, a simple calculation, and using the facts that  $\mu^2 \leq M/R$, and $\rho^2 \leq W/R$ shows that $\log (M\O \cdot U_1^2/\sigma^2) \leq C \log (WM)$. Using this together with \eqref{eq:Lem1-var}, and \eqref{eq:orlicz-bound}, and using $t = \beta \log(WM)$ in the Bernstein's inequality in Proposition \ref{prop:matbernpsi}, we have
\begin{align*}
\|\PT\setA_p^*\setA_p\PT-\PT\| &\leq C \max\bigg\{\sqrt{\left(\frac{ RW}{M\O} +\rho^2\frac{R}{\O}\right)}\sqrt{\beta\log(WM)}, ~\left(\frac{RW}{M\O}+\rho^2\frac{R}{\O}\right)\beta\log^2(WM)\bigg\}
\end{align*} 
We can conclude now that choosing $\O \geq C_\beta(\mu^2R(W/M)+ \rho^2 R)\log^2(WM)$ ensures that $\|\PT\setA_p^*\setA_p\PT-\PT\| \leq 1/2$, which proves the lemma after using the fact that $W \geq M$. 

\subsection{Proof of Lemma \ref{lem:JL}}\label{sec:JL}
Just as in the proof of Lemma \ref{lem:injectivity}, we will start with writing the $\setA_p^*\setA_p(\mW_{p-1})$ as a sum of independent random matrices using \eqref{eq:cAtcA} as follows
\begin{align*}
\setA_p^*\setA_p(\mW_{p-1})  = \underset{n \in \Gamma_p}{\sum}\sum_{\ell \in [\O]}  P\left[\va_n\va_n^*\mW_{p-1}\vd_{n\ell}\vd_{n\ell}^*\right].
\end{align*}
Recall that $\vd_{n\ell}$ are random binary defined earlier. Then the expectation of the random quantity above is
\begin{align*}
\E \setA_p^*\setA_p(\mW_{p-1})  &= \underset{n \in \Gamma_p}{\sum}\sum_{\ell \in [\O]}  P\E\left[\va_n\va_n^*\mW_{p-1}\vd_{\ell n}\vd_{\ell n}^*\right] = \underset{n \in \Gamma_p}{\sum}\sum_{\ell \in [\O]} P\left[\va_n\va_n^*\mW_{p-1} \mI_{\setB_\l} \right]= \mW_{p-1},
\end{align*}
where the last two equalities follow from the fact that 
\begin{equation}\label{eq:useful-facts-JL}
\sum_{n \in \Gamma_p}\va_n\va_n^* = \frac{1}{P}\mI_M,\quad\text{and} \quad \sum_{\ell \in [\O]} \E \vd_{n\ell}\vd_{n\ell}^* = \sum_{\ell \in [\O]} \mI_{\setB_\ell} = \mI_W. 
\end{equation}
We bound the operator norm $\|\setA_p^*\setA_p(\mW_{p-1})-\mW_{p-1}\|$. In light of discussion above,  $\setA_p^*\setA_p(\mW_{p-1})-\mW_{p-1}$  can be expressed as a following sum of independent, and zero mean random matrices
\[
\setA_p^*\setA_p(\mW_{p-1}) - \mW_{p-1} = \sum_{n,\ell} P\left[\va_n\va_n^*\mW_{p-1}\vd_{n\ell}\vd_{n\ell}^*-\E \va_n\va_n^*\mW_{p-1}\vd_{n\ell}\vd_{n\ell}^*\right],
\]
where $\sum_{n,\ell}$ is a shorthand for $\underset{n \in \Gamma_p}{\sum}\underset{\ell \in[\O]}{\sum}$. Define $
\mZ_{n\ell} := P\left[\va_n\va_n^*\mW_{p-1}\vd_{n\ell}\vd_{n\ell}^*-\E \va_n\va_n^*\mW_{p-1}\vd_{n\ell}\vd_{n\ell}^*\right].$ To compute the variance in \eqref{eq:matbernsigma}, we start with
\begin{align*}
&\sum_{n,\ell}\E \mZ_{n\ell}\mZ_{n\ell}^*= P^2\sum_{n,\ell}\big[\tfrac{W}{\O}\E \big( (\va_n^*\mW_{p-1}\vd_{n\ell})^2 \va_n\va_n^*\big)- \big(\E \va_n\va_n^*\mW_{p-1}\vd_{n\ell}\vd_{n\ell}^*\big)\big(\E \va_n\va_n^*\mW_{p-1}\vd_{n\ell}\vd_{n\ell}^*\big)^*\big],
\end{align*}
where we used the fact that $\|\vd_{n\ell}\|_2^2 = W/\O$. Since $\E\mZ_{n\ell}\mZ_{n\ell}^*$ is a symmetric positive-semidefinite matrix, that is, $\E \mZ_{n\ell}\mZ_{n\ell}^* \succcurlyeq \mathbf{0}$, this together with definition of $\mZ_{n\ell}$ implies that 
\[
\sum_{n,\ell}\big[\tfrac{W}{\O}\E \big( (\va_n^*\mW_{p-1}\vd_{n\ell})^2 \va_n\va_n^*\big)\big] \succcurlyeq \sum_{n,\ell}\big[\big(\E \va_n\va_n^*\mW_{p-1}\vd_{n\ell}\vd_{n\ell}^*\big)\big(\E \va_n\va_n^*\mW_{p-1}\vd_{n\ell}\vd_{n\ell}^*\big)^*\big],
\]
and, therefore, 
\begin{align}\label{eq:BVariance1}
\bigg\|\sum_{n,\ell}\E \mZ_{n\ell}\mZ_{n\ell}^*\bigg\| &\leq  P^2\frac{W}{\O}\bigg\|\sum_{n,\ell}\E \big( (\va_n^*\mW_{p-1}\vd_{n\ell})^2 \va_n\va_n^*\big)\bigg\|\leq P^2\frac{W}{\O}\sum_{n,\ell} \|\va_n^*\mW_{p-1}\mI_{\setB_\ell}\|_2^2 \va_n\va_n^*\notag\\
& \leq PW \max_{n \in [N]}\max_{\ell \in [\Omega]} \|\va_n^*\mW_{p-1}\mI_{\setB_\ell}\|_2^2 =  W \max_{m \in [M]}\max_{\ell \in [\Omega]} \|\ve_m^*\mW_{p-1}\mI_{\setB_\ell}\|_2^2 \notag\\
& \leq \nu_{p-1}^2 R \frac{W/M}{\Omega},
\end{align}
where the inequalities follow by using \eqref{eq:useful-facts-JL}, the definition of coherence $\nu_{p-1}^2$ in \eqref{eq:nup}, and \eqref{eq:a_n-e_m}.

For the second variance term in \eqref{eq:matbernsigma}, we skip through similar step as for the first term and land directly at 
\begin{align}\label{eq:interim-BVariance}
\bigg\|\underset{n \in \Gamma_p}{\sum}\sum_{\ell \in [\O]}\E \mZ_{n\ell}\mZ_{n\ell}^*\bigg\| & \leq \bigg\|\underset{n \in \Gamma_p}{\sum}\sum_{\ell \in [\O]}P^2 \|\va_n\|_2^2\E \big[ (\va_n^*\mW_{p-1}\vd_{n\ell})^2 \vd_{n\ell}\vd_{n\ell}^*\big]\bigg\|\notag\\
& = P \bigg\|\underset{n \in \Gamma_p}{\sum}\sum_{\ell \in [\O]} \E \big[ (\va_n^*\mW_{p-1}\vd_{n\ell})^2 \vd_{n\ell}\vd_{n\ell}^*\big]\bigg\|,
\end{align}
where the last equality is the result of \eqref{eq:a_n-e_m}. One can show that for a fixed vector $\vx \in \R^W$, and the fact that $\vd_{n\ell}$ is a vector with independent Rademacher random variables at locations indexed by $\setB_\ell$, and zero elsewhere, the following 
\begin{align}\label{eq:useful-fact2-JL}
\E \big[(\vx^*\vd_{n\ell})^2 \vd_{n\ell}\vd_{n\ell}^*\big]  \leq \|\vx_{\setB_\ell}\|_2^2 \mI_{\setB_\ell} + 2 \vx_{\setB_\ell}\vx_{\setB_\ell}^* \preccurlyeq 3\|\vx_{\setB_\ell}\|_2^2 \mI_{\setB_\ell}
\end{align}
holds, where $\vx_{\setB_\ell}$ is equal to $\vx$ on $\setB_\ell$, and zero elsewhere. Moreover, $\mI_{\setB_\ell}$ is a diagonal matrix with ones at $\setB_\ell$, and zero elsewhere. Using \eqref{eq:useful-fact2-JL} with $\vx^* = \va_n^*\mW_{p-1}$, we have 
\begin{align}\label{eq:BVariance2}
& \bigg\|\underset{n\in \Gamma_p}{\sum}\sum_{\ell \in [\O]}\E \mZ_{n\ell}\mZ_{n\ell}^*\bigg\|\leq 3P\bigg\|\underset{n \in \Gamma_p}{\sum}\sum_{\ell \in [\O]} \|\va_n^*\mW_{p-1}\mI_{\setB_\ell}\|_2^2\mI_{\setB_\ell}\bigg\|\notag\\
&\qquad \leq 3P \underset{n \in \Gamma_p}{\sum}\max_{\ell \in [\Omega]} \|\va_n^*\mW_{p-1}\mI_{\setB_\ell}\|_2^2 \leq 3PM\max_{n \in \Gamma_p}\max_{\ell \in [\O]} \|\va_n^*\mW_{p-1}\mI_{\setB_\ell}\|_{2}^2 \leq 3 \nu_{p-1}^2 \frac{R}{\Omega},
\end{align}
where in the last inequality, we use the definition of $\nu_{p-1}^2$ in \eqref{eq:coherence-bound} combined with \eqref{eq:a_n-e_m}. In light of \eqref{eq:matbernsigma}, the maximum of  \eqref{eq:BVariance1}, and \eqref{eq:BVariance2} accounts for the variance $\sigma^2$
\begin{align}\label{eq:B-Variance}
\sigma^2 \leq 3\nu_{p-1}^2R \frac{(W/M)+1}{\Omega} \leq 6 \nu_{p-1}^2R \frac{(W/M)}{\Omega},
\end{align}
where in the last inequality follows from our assumption that $ W \geq M$. Finally, we need to compute an upper bound on the Orlicz norm of the random variable $\|\mZ_{n\ell}\|$. Begin by using similar simple facts above that 
\begin{align}\label{eq:Zij-operatornorm}
\|\mZ_{n\ell}\| &\leq 2P \|\va_n\va_n^*\mW_{p-1}\vd_{n\ell}\vd_{n\ell}^*\| =  2P \|\va_n\|_2\|\vd_{n\ell}\|_2 |\va_n^*\mW_{p-1}\vd_{n\ell}|= 2P\frac{1}{\sqrt{P}}\sqrt{\frac{W}{\Omega}}|\va_n^*\mW_{p-1}\vd_{n\ell}|.
\end{align}
Using standard calculations; see, for example, \cite{vershynin10in}, we can compute the following finite bound on the Orlicz-1 norm of the random variable $|\va_n^*\mW_{p-1}\vd_{n\ell}|$
\begin{align*}
\max_{n \in \Gamma_p}\max_{\ell \in [\Omega]}\|\va_n^*\mW_{p-1}\vd_{n\ell}\|_{\psi_2}& = \max_{m \in [M]} \max_{\ell \in [\Omega]}\frac{1}{P}\|\ve_m^*\mW_{p-1}\vd_{n\ell}\|_{\psi_2}\\
&\leq \frac{C}{P}\max_{n \in \Gamma_p}\max_{\ell \in [\Omega]}\|\ve_m^*\mW_{p-1}\mI_{\setB_\ell}\|_{2}\leq C\nu_{p-1}\sqrt{\frac{R}{PM\O}},
\end{align*}
where the last inequality follows from \eqref{eq:nup}. Using \eqref{eq:Zij-operatornorm}, $P=N/M$, and \eqref{eq:matpsinorm} then directly gives us 
\begin{align}\label{eq:B-Orlicz}
U^2_2 := \max_{n \in \Gamma_p}\max_{\ell \in [\Omega]} \|\mZ_{n\ell}\|^2_{\psi_2} \leq C\nu^2_{p-1} R\frac{W/M}{\O^2}.
\end{align}
Moreover using a loose bound on variance $\sigma^2 \leq 3\nu_{p-1}^2R ((W/M)+1)/\Omega$, it is easy to see that
\[
\log\bigg[\frac{\O M\cdot U_2^2}{\sigma^2}\bigg] \leq C\log M.
\]
The results in \eqref{eq:BVariance1}, \eqref{eq:BVariance2}, and \eqref{eq:B-Orlicz} can be plugged in Proposition \ref{prop:matbernpsi} to obtain 
\begin{align}
&\|\setA_p^*\setA_p(\mW_{p-1})-\mW_{p-1}\|_{\F} \leq C\max\left\{\sqrt{\nu_{p-1}^2 R\frac{W/M}{\O}}\sqrt{\beta\log(WM)},\sqrt{\nu_{p-1}^2R\frac{W/M}{\O^2}}\beta\log^{3/2}(WM)\right\}\label{eq:JL-bound}
\end{align}
with $t = (\beta-1) \log(WM)$, which holds with probability at least $1-\mathcal{O}(WM)^{-\beta}$. Recall that $W \geq M$. The lemma now follows by using the bound on $\nu_{p-1}^2$ in \eqref{eq:coherence-bound}, and choosing $\Omega \geq C_\beta R \nu^2 (W/M) \log^{3/2}W$ for a universal constant $C_\beta$ that only depends on a fixed parameter $\beta \geq 1$.
\section{Proof of Theorem \ref{thm:stable-rec}}\label{sec:Stability}
The first step in the proof is the following Oracle inequality in \cite{koltchinskii10nu} that gives an upper bound on the deviation of $\widehat{\mX}$ in \eqref{eq:KLT-solution} from the true solution $\mX_0$ in the mean squared sense. 
\begin{thm}[Oracle inequlaity in \cite{koltchinskii10nu}]\label{thm:interim-stability-result}
	Suppose we observe the noisy measurements $\vy$ in \eqref{eq:noisy-measurements} of $\mX_0$ with $\text{rank}(\mX_0) \leq R$, and  it is given that $\|\setA^*(\vy)-\E \setA^*(\vy)\|\leq \lambda/2$ fro some scalar $\lambda \geq 0$. Then the solution $\widehat{\mX}$ of the nuclear norm penalized estimator in \eqref{eq:KLT-estimator} obeys $\|\widehat{\mX} -\mX_0\|_{F}^2 \leq \min \big(2\lambda\|\mX_0\|_*,1.46\lambda^2R\big).$
\end{thm}
All that is required is to bound the spectral norm:
\begin{align}\label{eq:stable-rec-ingdt}
\|\setA^*(\vy)-\E \setA^*(\vy)\| & \leq \|\setA^*\setA(\mX_0)-\mX_0\| + \|\setA^*(\vxi)\|
\end{align}
We begin by bounding the first term above $\|\setA^*\setA(\mX_0)-\mX_0\|$ using a corollary to Lemma \ref{lem:JL} stated as follows. 
\begin{cor}\label{cor:JL}
	Let $\mX_0$ be a fixed $M \times W$ matrix defined in \ref{eq:X0def} then 
	\begin{align}
	&\|\setA^*\setA(\mX_0)-\mX_0\| \leq C\|\mX_0\|_{\F}\max\left\{\sqrt{\nu^2\frac{W/M}{\O }}\sqrt{\beta\log W},\sqrt{\nu^2\frac{W/M}{\O^2}}\beta\log^{3/2}W\right\}\label{eq:JL-corllary-bound}.
	\end{align}
	with probability at least $1-\setO(W^{-\beta})$.
\end{cor}
\begin{proof}
	The proof of the corollary is very similar to the proof of Lemma \ref{lem:JL}--- the main difference is that  the number of partitions is  $P = N/M = 1$. Moreover, we have $\mX_0$ in place of $\mW_{p-1}$, and in the proof development replace $\|\mW_{p-1}\|_{\F} \leq 2^{-p+1}\sqrt{R}$ with $\|\mX_0\|_{\F}$ to obtain bound in \eqref{eq:JL-corllary-bound}, which is understandably similar to \eqref{eq:JL-bound}
\end{proof}
\begin{lem}\label{lem:stable-rec}
	Fix $\beta \geq 1$. The for a sufficiently large constant $C$, the following bound
	\begin{align*}
	\|\setA^*(\vxi)\| \leq C\|\vxi\|_{\psi_2}\sqrt{\frac{W/M}{\O}}\sqrt{\beta\log W}
	\end{align*}
	holds with probability at least $1-\setO(W^{-\beta}).$
\end{lem}
Using Corollary \ref{cor:JL}, and Lemma \ref{lem:stable-rec}, we can bound \eqref{eq:stable-rec-ingdt}, and obtain 
\[
\lambda \geq C\sqrt{ \frac{(W/M)(\nu^2\|\mX_0\|_{\F}^2 + \|\vxi\|_{\psi_2}^2)}{\O}}\sqrt{\beta\log^{3/2} W}
\]
with probability at least $1-\setO(W^{-\beta})$. Taking $\|\mX_0\|_{\F} = 1$ without loss of generality, and $\O \geq C_\beta\nu^2R\tfrac{W}{M}\log^{3/2}W$, where $C_\beta$ is a universal constant that depends on a fixed parameter $\beta \geq 1$, allows us to choose $\lambda \geq \sqrt{\|\vxi\|_{\psi_2}^2/R}$. With this, an application of Theorem \ref{thm:interim-stability-result} proves Theorem \ref{thm:stable-rec}.

\subsection{Proof of Lemma \ref{lem:stable-rec}}
The proof of this lemma requires the use of matrix Bernstein's inequality \ref{prop:matbernpsi}. 
As it is required to bound the spectral norm of the sum $\setA^*(\vxi) = \sum_{n,\ell} \xi_{n}[\ell]\va_n\vd^*_{n\ell}$, we start with the summands $\mZ_{n\ell} = \xi_n[\ell]\va_n\vd^*_{n\ell}$. Because variables $\xi_{n\ell}$ are zero mean, it follows that $\E\mZ_{n\ell} = \mathbf{0}$. We start by computing the variance 
\begin{align}\label{eq:strec-var1}
\bigg\|\sum_{n,\ell} \E\mZ_{n\ell}\mZ_{n\ell}^*\bigg\| &= \frac{W}{\Omega}\bigg\|\sum_{n,\ell} \E |\xi_{n}[\ell]|^2 \cdot \va_n\va_n^*\bigg\| = \frac{W}{\Omega}\bigg\| \max_{n \in [N]}\sum_{\ell \in [\Omega]} \E |\xi_n[\ell]|^2 \cdot \sum_{n \in [N]} \va_n\va_n^* \bigg\| \leq \frac{W}{N\Omega}\|\vxi\|_{\psi_2}^2,
\end{align}
where the last inequality follows from the facts that $\sum_{n \in [N]} \va_n\va_n^* = \mI_M$, and that  $\xi_n[\ell]$ for $(n,\ell) \in [N]\times[\Omega]$ are independent and identically distributed implying $\max_{n \in [N]}\sum_{\ell \in [\Omega]} \E |\xi_n[\ell]|^2 \leq C \|\vxi\|_{\psi_2}^2/N$.

Similarly arguments lead to 
\begin{align}\label{eq:strec-var2}
\bigg\|\sum_{n,\ell} \E\mZ_{n\ell}\mZ_{n\ell}^*\bigg\| &=  \bigg\|\sum_{n,\ell}\|\va_n\|_2^2 \E |\xi_n[\ell]|^2 \E \vd_{n\ell}\vd_{n\ell}^*\bigg\| = \frac{M}{N} \bigg\|\sum_{n,\ell} \E |\xi_n[\ell]|^2 \mI_{\setB_\ell}\bigg\| \leq \frac{M}{N\Omega}\|\vxi\|_{\psi_2}^2.
\end{align}
Combining \eqref{eq:strec-var1} and \eqref{eq:strec-var2} and using \eqref{eq:matbernsigma} gives $\sigma^2 \leq  \|\vxi\|^2_{\psi_2}PW/M\Omega,$ where $P = N/M$, and we assume that $W\geq M$. 
The final quantity required is the Orlicz norm of  $\mZ_{n\ell}$, which is simply
\begin{align*}
\|\mZ_{n\ell}\|^2_{\psi_2} & = \|\xi_n[\ell]\|_{\psi_2}^2\|\va_n\vd_{n\ell}^*\|^2 = C \frac{1}{N\O}\|\vxi\|_{\psi_2}^2\cdot \frac{MW}{N\Omega} = C\|\vxi\|_{\psi_2}^2 \frac{1}{P^2} \frac{W}{M\Omega^2}
\end{align*}
then 
\[
\|\mZ_{n\ell}\|_{\psi_2} \log^{1/2}\left(\frac{M\O\cdot\|\mZ_{n\ell}\|_{\psi_2}^2}{\sigma^2}\right) \leq C \sqrt{\|\vxi\|_{\psi_2}^2\frac{W/M}{P^2\O^2}}\log^{1/2} (MW).
\]
At the end, using $t = \beta\log W$ in the Bernstein's bound \eqref{eq:matbernpsi}, we have 
\[
\|\setA^*(\vxi)\| \leq C\max \left\{\|\vxi\|_{\psi_2}\sqrt{\frac{W/M}{\O}}\sqrt{\beta\log (MW)}, \|\vxi\|_{\psi_2}\sqrt{\frac{W/M}{P^2\O^2}}(\beta\log^{3/2} (MW))\right\},
\]
and using the fact that $P = \setO(\log W)$, and $M \leq W$ from \eqref{eq:P-bound} proves the result.

\bibliographystyle{IEEEtran}
\bibliography{IEEEabrv,DemodArray-references}


%
%
%


\end{document}